\documentclass[11pt]{article}

\usepackage[paper=letterpaper, margin=1in]{geometry}

\usepackage{ifthen}

\usepackage{dsfont}

\ifthenelse{\isundefined{\GenerateShortVersion}}
{
\def\LongVersion{}
\def\LongVersionEnd{}
\long\def\ShortVersion#1\ShortVersionEnd{}
}
{
\def\ShortVersion{}
\def\ShortVersionEnd{}
\long\def\LongVersion#1\LongVersionEnd{}
}

\newcommand{\Ignore}[1]{\ignorespaces}

\usepackage{amsmath, amssymb}

\usepackage{amsthm}

\usepackage{ifpdf}

\usepackage{authblk}

\ifpdf
\usepackage[pdftex]{graphicx}
\else
\usepackage[dvips]{graphicx}
\fi

\usepackage[%
  ocgcolorlinks,%
  linkcolor=blue,%
  filecolor=blue,%
  citecolor=blue,%
  urlcolor=blue]{hyperref}

\usepackage{enumitem}

\usepackage{xcolor}

\usepackage{mdframed}

\usepackage{subfigure}
\usepackage{tikz}
\usetikzlibrary{backgrounds,automata}

\LongVersion 
\LongVersionEnd 

\ShortVersion 
\usepackage{times}
\ShortVersionEnd 

\ShortVersion 
\renewcommand{\paragraph}[1]{\par\noindent\textbf{#1}}
\ShortVersionEnd 

\newtheorem{theorem}{Theorem}[section]
\newtheorem{lemma}[theorem]{Lemma}
\newtheorem{observation}[theorem]{Observation}
\newtheorem{corollary}[theorem]{Corollary}

\theoremstyle{definition}
\newtheorem*{definition*}{Definition}
\theoremstyle{plain}

\LongVersion 
\newenvironment{MathMaybe}[0]
{\begin{displaymath}\ignorespaces}
{\end{displaymath}}
\LongVersionEnd 
\ShortVersion 

\ShortVersionEnd 

\newcommand{\Integers}{\mathbb{Z}}
\newcommand{\Neighbors}{\mathit{N}}
\newcommand{\Faulty}[1]{\widehat{#1}}
\newcommand{\Able}[1]{\overline{#1}}
\renewcommand{\Pr}{\mathbb{P}}
\newcommand{\Level}{\mathcal{L}}
\newcommand{\Distance}{\operatorname{dist}}
\newcommand{\Diameter}{\operatorname{diam}}
\newcommand{\forward}{\phi}
\newcommand{\outwards}{\psi}
\newcommand{\Outwards}{\Psi}
\newcommand{\ProtectedV}{\mathit{V}_{\mathrm{p}}}
\newcommand{\ProtectedE}{\mathit{E}_{\mathrm{p}}}
\newcommand{\OutProtectedV}{\mathit{V}_{\mathrm{op}}}
\newcommand{\Sign}{\operatorname{sign}}
\newcommand{\Degree}{\operatorname{deg}}

\newcommand{\AlgAU}{\ensuremath{\mathtt{AlgAU}}}
\newcommand{\AlgLE}{\ensuremath{\mathtt{AlgLE}}}
\newcommand{\AlgMIS}{\ensuremath{\mathtt{AlgMIS}}}
\newcommand{\ModuleRestart}{\ensuremath{\mathtt{Restart}}}
\newcommand{\ModuleRandCount}{\ensuremath{\mathtt{RandCount}}}
\newcommand{\ModuleElect}{\ensuremath{\mathtt{Elect}}}
\newcommand{\ModuleDetectLE}{\ensuremath{\mathtt{DetectLE}}}
\newcommand{\ModuleRandPhase}{\ensuremath{\mathtt{RandPhase}}}
\newcommand{\ModuleCompete}{\ensuremath{\mathtt{Compete}}}
\newcommand{\ModuleDetectMIS}{\ensuremath{\mathtt{DetectMIS}}}

\newcommand{\varFlag}{\mathtt{flag}}
\newcommand{\varCandidate}{\mathtt{candidate}}
\newcommand{\varStep}{\mathtt{step}}
\newcommand{\InMIS}{\mathit{IN}}
\newcommand{\OutMIS}{\mathit{OUT}}

\newcommand{\Sect}{Sec.}
\newcommand{\Thm}{Thm.}
\newcommand{\Lem}{Lem.}
\newcommand{\Obs}{Obs.}

\begin{document}

\title{A Thin Self-Stabilizing Asynchronous Unison Algorithm with Applications
to Fault Tolerant Biological Networks}

\author{Yuval Emek\footnote{%
The work of Y.~Emek was supported by an Israeli Science Foundation grant
number 1016/17.}}
\affil{Technion --- Israel Institute of Technology. \\
\texttt{yemek@technion.ac.il}}
\author{Eyal Keren}
\affil{Technion --- Israel Institute of Technology. \\
\texttt{eyal.keren@campus.technion.ac.il}}

\date{}

\maketitle

\begin{abstract}
Introduced by Emek and Wattenhofer (PODC 2013), the \emph{stone age (SA)}
model provides an abstraction for network algorithms distributed over
randomized finite state machines.
This model, designed to resemble the dynamics of biological processes in
cellular networks, assumes a weak communication scheme that is built upon the
nodes' ability to sense their vicinity in an asynchronous manner.
Recent works demonstrate that the weak computation and communication
capabilities of the SA model suffice for efficient solutions to some core
tasks in distributed computing, but they do so under the (somewhat less
realistic) assumption of fault free computations.
In this paper, we initiate the study of \emph{self-stabilizing} SA algorithms
that are guaranteed to recover from any combination of transient faults.
Specifically, we develop efficient self-stabilizing SA algorithms for the
\emph{leader election} and \emph{maximal independent set} tasks in bounded
diameter graphs subject to an asynchronous scheduler.
These algorithms rely on a novel efficient self-stabilizing \emph{asynchronous
unison (AU)} algorithm that is ``thin'' in terms of its state space:
the number of states used by the AU algorithm is linear in the graph's
diameter bound, irrespective of the number of nodes.
\end{abstract}

\thispagestyle{empty}
\clearpage
\setcounter{page}{1}

\section{Introduction}
\label{section:introduction}
A fundamental dogma in distributed computing is that a distributed algorithm
cannot be deployed in a real system unless it can cope with faults.
When it comes to recovering from transient faults, the agreed upon
concept for fault tolerance is \emph{self-stabilization}.
Introduced in the seminal paper of Dijkstra \cite{Dijkstra1982}, an algorithm
is self-stabilizing if it is guaranteed to converge to a correct output from
any (possibly faulty) initial configuration
\cite{dolev_self-stabilization_2000, self_2019}.

Similarly to distributed man-made digital systems, self-stabilization is also
crucial to the survival of biological distributed systems.
Indeed, these systems typically lack a central component that can determine
the initial system configuration in a coordinated manner and more often than
not, they are exposed to environmental conditions that may lead to transient
faults.
On the other hand, biological distributed systems are often inferior to
man-made distributed systems in terms of the computation and communication
capabilities of their individual components (e.g., a single cell in an organ),
thus calling for a different model of distributed network algorithms.

Aiming to capture distributed processes in biological cellular networks, Emek
and Wattenhofer \cite{EmekW2013stone} introduced the \emph{stone age (SA)}
model that provides an abstraction for distributed algorithms in a network of
randomized finite state machines that communicate with their network neighbors
using a fixed message alphabet based on a weak communication scheme.
Since then, the power and limitations of distributed SA algorithms have been
studied in several papers.
In particular, it has been established that some of the most fundamental tasks
in the field of distributed graph algorithms can be solved efficiently in
this restricted model
\cite{EmekW2013stone, AfekEK2018selecting, AfekEK2018synergy, EmekU2020dynamic}.
However, for the most part, the existing literature on the SA model focuses on
fault free networks and little is known about self-stabilizing distributed
algorithms operating under this model.\footnote{%
In \cite{EmekU2020dynamic}, Emek and Uitto study the SA model in networks that
undergo dynamic topology changes, including node deletion that may be seen as
(permanent) crash failures.}

In the current paper, we strive to change this situation:
Focusing on graphs of \emph{bounded diameter}, we design efficient
self-stabilizing SA algorithms for leader election and maximal independent set
--- two of the most fundamental and extensively studied tasks in the theory
of distributed computing.
A key technical component in the algorithms we design is a self-stabilizing
\emph{synchronizer} for SA algorithms in graphs of bounded diameter.
This synchronizer relies on a novel anonymous size-uniform self-stabilizing
algorithm for the \emph{asynchronous unison} task
\cite{CouvreurFG1992asynchronous, AwerbuchKMPV1993time}
that operates with a number of states linear in the graphs diameter bound
$D$.
To the best of our knowledge, this is the first self-stabilizing asynchronous
unison algorithm for graphs of general topology whose state space is expressed
solely as a function of $D$, independently of the number $n$ of nodes.

The decision to focus on bounded diameter graphs is motivated by regarding
this graph family as a natural extension of complete graphs.
Indeed, environmental obstacles may disconnect (permanently or temporarily)
some links in an otherwise fully connected network, thus increasing its
diameter beyond one, but hopefully not to the extent of exceeding a certain
fixed upper bound.
Fully connected networks go hand in hand with broadcast communication that
prevail in the context of both man-made (e.g., contention resolution
in multiple access channels) and biological (e.g., quorum sensing in bacterial
populations) distributed processes.
As the SA model offers a (weak form) of broadcast communication, it makes
sense to investigate its power and limitations in such networks and their
natural extensions.

\subsection{Computational Model}
\label{section:intro:model}
The computational model used in this paper is a simplified version of the
\emph{stone age (SA)} model of Emek and Wattenhofer \cite{EmekW2013stone}.
This model captures anonymous size-uniform distributed algorithms with bounded
memory nodes that exchange information by means of an asynchronous variant of
the \emph{set-broadcast} communication scheme (cf.\
\cite{HellaJKLLLSV2015weak}) with no \emph{sender collision detection} (cf.\
\cite{AfekABCHK2011beeping}).
Formally, given a distributed task $\mathcal{T}$ defined over a set
$\mathcal{O}$ of output values, an algorithm $\Pi$ for $\mathcal{T}$ is
encoded by the $4$-tuple
$\Pi
=
\left\langle
Q, Q_\mathcal{O}, \omega, \delta
\right\rangle$,
where
\begin{itemize}

\item
$Q$ is a set of states;

\item
$Q_\mathcal{O} \subseteq Q$
is a set of output states;

\item
$\omega : Q_\mathcal{O} \to \mathcal{O}$
is a surjective function that maps each output state to an output value;
and

\item
$\delta : Q \times \left\lbrace 0, 1 \right\rbrace^{Q} \to 2^{Q}$
is a state transition function (to be explained soon).\footnote{%
The notation $2^{Q}$ denotes the power set of $Q$.}

\end{itemize}
We would eventually require that the \emph{state space} of $\Pi$, namely, the
size $|Q|$ of the state set, is fixed, and in particular independent of the
graph on which $\Pi$ runs, as defined in \cite{EmekW2013stone}.
To facilitate the discussion though, let us relax this requirement for the
time being.

Consider a finite connected undirected graph
$G = (V, E)$.
A \emph{configuration} of $G$ is a function
$\mathcal{C} : V \to Q$
that determines the state
$\mathcal{C}(v) \in Q$
of node $v$ for each
$v \in V$.
We say that a node
$v \in V$
\emph{senses} state
$q \in Q$
under $\mathcal{C}$ if there exists some (at least one) node
$u \in \Neighbors^{+}(v)$
such that
$\mathcal{C}(u) = q$.\footnote{%
Throughout this paper, we denote the neighborhood of a node $v$ in $G$ by
$\Neighbors(v) = \{ u \in V \mid (u, v) \in E \}$
and the inclusive neighborhood of $v$ in $G$ by
$\Neighbors^{+}(v) = \Neighbors(v) \cup \{ v \}$.}
The \emph{signal} of $v$ under $\mathcal{C}$ is the binary
vector
$\mathcal{S}_{v}^{\mathcal{C}} \in \{ 0, 1 \}^{Q}$
defined so that
$\mathcal{S}_{v}^{\mathcal{C}}(q) = 1$
if and only if $v$ senses state
$q \in Q$;
in other words, the signal of node $v$ allows $v$ to determine for each state
$q \in Q$
whether $q$ appears in its (inclusive) neighborhood, but it does not allow $v$
to count the number of such appearances, nor does it allow $v$ to identify the
neighbors residing in state $q$.

The execution of $\Pi$ progresses in discrete \emph{steps}, where step
$t \in \Integers_{\geq 0}$
spans the time interval
$[t, t + 1)$.
Let
$\mathcal{C}^{t} : V \to Q$
be the configuration of $G$ at time $t$ and let
$\mathcal{S}_{v}^{t} = \mathcal{S}_{v}^{\mathcal{C}^{t}}$
denote the signal of node
$v \in V$
under $\mathcal{C}^{t}$.
We consider an \emph{asynchronous} schedule defined by means of a sequence of
node \emph{activations} (cf.\ a distributed fair daemon
\cite{DuboisT2011taxonomy}).
Formally, a malicious adversary, who knows $\Pi$ but is oblivious to the
nodes' coin tosses, determines the initial configuration $\mathcal{C}^{0}$ and
a subset
$A^{t} \subseteq V$
of nodes to be activated at time $t$ for each
$t \in \Integers_{\geq 0}$.
If node
$v \in V$
is not activated at time $t$, then
$\mathcal{C}^{t + 1}(v) = \mathcal{C}^{t}(v)$.
Otherwise
($v \in A^{t}$),
the state of $v$ is updated in step $t$ from
$\mathcal{C}^{t}(v)$
to
$\mathcal{C}^{t + 1}(v)$
picked uniformly at random from
$\delta \left( \mathcal{C}^{t}(v), \mathcal{S}_{v}^{t} \right)$.
We emphasize that all nodes
$v \in V$
obey the same state transition function $\delta$.

Fix some schedule
$\{ A^{t} \}_{t \geq 0}$.
The adversary is required to prevent ``node starvation'' in the sense that
each node must be activated infinitely often.
Given a time
$t \in \Integers_{\geq 0}$,
let $\varrho(t)$ be the earliest time satisfying the property that for
every node
$v \in V$,
there exists a time
$t \leq t' < \varrho(t)$
such that
$v \in A_{t'}$.
This allows us to introduce the \emph{round operator} $\varrho^{i}(t)$ defined
by setting
$\varrho^{0}(t) = t$
and
$\varrho^{i}(t) = \varrho \left( \varrho^{i - 1}(t) \right)$
for
$i = 1, 2, \dots$
Denote
$R(i) = \varrho^{i}(0)$
for
$i = 0, 1, \dots$,
and observe that if
$R(i) \leq t < R(i + 1)$,
then
$R(i + 1) \leq \varrho(t) < R(i + 2)$.

A configuration
$\mathcal{C} : V \to Q$
is said to be an \emph{output configuration} if
$\mathcal{C}(v) \in Q_{\mathcal{O}}$
for every
$v \in V$,
in which case, we regard $\omega(\mathcal{C}(v))$ as the output of node $v$
under $\mathcal{C}$ and refer to
$\omega \circ \mathcal{C}$
as the \emph{output vector} of $\mathcal{C}$.
We say that the execution of $\Pi$ on $G$ has \emph{stabilized} by time
$t \in \Integers_{\geq 0}$
if
(1)
$\mathcal{C}^{t'}$ is an output configuration for every
$t' \geq t$;
and
(2)
the output vector sequence
$\{ \omega \circ \mathcal{C}^{t'} \}_{t' \geq t}$
satisfies the requirements of the distributed task $\mathcal{T}$ for which
$\Pi$ is defined (the requirements of the distributed tasks studied in the
current paper are presented in \Sect{}~\ref{section:intro:tasks}).

The algorithm is \emph{self-stabilizing} if for any choice of initial
configuration $\mathcal{C}^{0}$ and schedule
$\{ A^{t} \}_{t \geq 0}$,
the probability that $\Pi$ has stabilized by time $R(i)$ goes to $1$ as
$i \to \infty$.
We refer to the smallest $i$ for which the execution has stabilized by time
$R(i)$ as the \emph{stabilization time} of this execution.
The stabilization time of a randomized (self-stabilizing) algorithm on a given
graph is a random variable and one typically aims towards bounding it in
expectation and whp.\footnote{%
In the context of a randomized algorithm running on an $n$-node graph, we say
that event $A$ occurs with high probability, abbreviated whp, if
$\Pr (A) \geq 1 - n^{-c}$
for an arbitrarily large constant $c$.}

The schedule
$\{ A^{t} \}_{t \geq 0}$
is said to be \emph{synchronous} if
$A^{t} = V$
for all
$t \in \Integers_{\geq 0}$
which means that
$R(i) = i$
for
$i = 0, 1, \dots$
A (self-stabilizing) algorithm whose correctness and stabilization time
guarantees hold under the assumption of a synchronous schedule is called a
\emph{synchronous algorithm}.
We sometime emphasize that an algorithm does not rely on this assumption by
referring to it as an \emph{asynchronous algorithm}.

\subsection{Distributed Tasks}
\label{section:intro:tasks}
In this paper, we focus on three classic (and extensively studied) distributed
tasks, defined over a finite connected undirected graph
$G = (V, E)$.
In the first task, called \emph{asynchronous unison (AU)}
\cite{CouvreurFG1992asynchronous} (a.k.a.\ \emph{distributed pulse}
\cite{AwerbuchKMPV1993time}), each node in $V$ outputs a \emph{clock value}
taken from an (additive) cyclic group $K$.
The task is then defined by the following two conditions:
The \emph{safety} condition requires that if two neighboring nodes output
clock values
$\kappa \in K$
and
$\kappa' \in K$,
then
$\kappa' \in \{ \kappa - 1, \kappa, \kappa + 1 \}$,
where the $+1$ and $-1$ operations are with respect to $K$.
The \emph{liveness} condition requires that for every (post stabilization)
time $t$ and for every
$i \in \Integers_{> 0}$,
each node updates its clock value at least $i$ times during the time interval
$[t, \varrho^{\Diameter(G) + i}(t))$,
where $\Diameter(G)$ denotes the diameter of $G$;
these updates are performed by and only by applying the $+1$ operation of
$K$.

The other two distributed tasks considered in this paper are \emph{leader
election (LE)} and \emph{maximal independent set (MIS)}.
Both tasks are defined over a binary set
$\mathcal{O} = \{ 0, 1 \}$
of output values and are static in the sense that once the algorithm has
stabilized, its output vector remains fixed.
In LE, it is required that exactly one node in $V$ outputs $1$;
in MIS, it is required that the set
$U \subseteq V$
of nodes that output $1$ is independent, i.e.,
$(U \times U) \cap E = \emptyset$,
whereas any proper superset of $U$ is not independent.
We note that LE and MIS correspond to global and local \emph{mutual
exclusion}, respectively, and that the two tasks coincide if $G$ is the
complete graph.

\subsection{Contribution}
\label{section:intro:contribution}
In what follows, we refer to the class of graphs whose diameter is up-bounded
by $D$ as \emph{$D$-bounded diameter}.
Our first result comes in the form of developing a new self-stabilizing AU
algorithm.

\begin{theorem} \label{theorem:AU}
The class of $D$-bounded diameter graphs admits a deterministic
self-stabilizing AU algorithm that operates with state space
$O (D)$
and stabilizes in time
$O (D^{3})$.
\end{theorem}

To the best of our knowledge, the algorithm promised in
\Thm{}~\ref{theorem:AU} is the first self-stabilizing AU algorithm for
general graphs
$G = (V, E)$
with state space linear in the diameter bound $D$, irrespective of any other
graph parameter including
$n = |V|$.
This remains true even when considering algorithms designed to work under much
stronger computational models (see
\Sect{}~\ref{section:related-work-discussion} for further discussion).
Moreover, to the best of our knowledge, this is also the first anonymous
size-uniform self-stabilizing algorithm for the AU task whose stabilization
time is expressed solely as a (polynomial) function of $D$, again, irrespective
of $n$.
Expressing the guarantees of AU algorithms with respect to $D$ is advocated
given the central role that the diameter of $G$ plays in the liveness
condition of the AU task.

There is a well known reduction from the problem of network synchronization
(a.k.a.\ \emph{synchronizer} \cite{Awerbuch1985complexity}) to AU under
computational models that support unicast communication (see, e.g.,
\cite{AwerbuchKMPV1993time}).
A similar reduction can be established also for our weaker computational
model, yielding the following corollary.

\begin{corollary} \label{corollary:synchronizer}
Suppose that a distributed task $\mathcal{T}$ admits a synchronous
self-stabilizing algorithm that on $D$-bounded diameter $n$-node graphs,
operates with state space $g(D)$ and stabilizes in time at most
$f(n, D)$
in expectation and whp.
Then, $\mathcal{T}$ admits an asynchronous self-stabilizing algorithm that on
$D$-bounded diameter $n$-node graphs, operates with state space
$O (D \cdot (g(D))^{2})$
and stabilizes in time at most
$f(n, D) + O (D^{3})$
in expectation and whp.
\end{corollary}

Next, we turn our attention to LE and MIS and develop efficient
self-stabilizing asynchronous algorithms for these tasks by combining
Corollary~\ref{corollary:synchronizer} with the following two theorems.

\begin{theorem} \label{theorem:LE}
There exists a synchronous self-stabilizing LE algorithm that on $D$-bounded
diameter $n$-node graphs, operates with state space
$O (D)$
and stabilizes in time
$O (D \cdot \log n)$
in expectation and whp.
\end{theorem}

\begin{theorem} \label{theorem:MIS}
There exists a synchronous self-stabilizing MIS algorithm that on $D$-bounded
diameter $n$-node graphs, operates with state space
$O (D)$
and stabilizes in time
$O ((D + \log n) \log n)$
in expectation and whp.
\end{theorem}

We emphasize that when the diameter bound $D$ is regarded as a fixed
parameter, the state space of our algorithms reduces to a constant, as
required in the SA model \cite{EmekW2013stone}.
In this case, the asymptotic stabilization time bounds in \Thm{}
\ref{theorem:AU}, \ref{theorem:LE}, and \ref{theorem:MIS}
should be interpreted as
$O (1)$,
$O (\log n)$,
and
$O (\log^{2} n)$,
respectively.

\subsection{Paper's Outline}
\label{section:paper-outline}
The remainder of this paper is organized as follows.
In \Sect{}~\ref{section:AU}, we develop our self-stabilizing AU algorithm and
establish \Thm{}~\ref{theorem:AU}.
The self-stabilizing synchronous LE and MIS algorithms promised in \Thm{}
\ref{theorem:LE} and \ref{theorem:MIS}, respectively, are presented in
\Sect{}~\ref{section:algorithms-LE-MIS}.
\Sect{}~\ref{section:synchronizer} is dedicated to a SA variant of the well
known reduction from self-stabilizing network synchronization to the AU task,
establishing Corollary~\ref{corollary:synchronizer}.
We conclude with additional related literature and a discussion of the place
of our work within the scope of the existing ones;
this is done in \Sect{}~\ref{section:related-work-discussion}.

\section{Asynchronous Unison}
\label{section:AU}
In this section, we establish \Thm{}~\ref{theorem:AU} by introducing a
deterministic self-stabilizing algorithm called \AlgAU{} for the AU task on
$D$-bounded diameter graphs, whose state space and stabilization time are
bounded by
$O (D)$
and
$O (D^{3})$,
respectively.
The algorithm is presented in \Sect{}~\ref{section:AU:construction} and
analyzed in \Sect{}~\ref{section:AU:analysis}.
Before diving into the technical parts, \Sect{}~\ref{section:AU:overview}
provides a short overview of \AlgAU{}'s design principles, and how they
compare with existing constructions.

\subsection{Technical Overview}
\label{section:AU:overview}
Most existing efficient constructions of self-stabilizing AU algorithms with
bounded state space rely on some sort of a reset mechanism.
This mechanism is invoked upon detecting an illegal configuration that usually
means a ``clock discrepancy'', namely, graph neighbors whose states are
associated with non-adjacent clock values of the acyclic group $K$.
The reset mechanism is designed so that it brings the system back to a legal
configuration, from which a fault free execution can proceed.
It turns out though that designing a self-stabilizing AU algorithm with state
space
$O (D)$
based on a reset mechanism is more difficult than what one may have expected
as demonstrated by the failed attempt presented in
Appendix~\ref{appendix:failed-attempt}.

Discouraged by this failed attempt, we followed a different approach and
designed our self-stabilizing AU algorithm without a reset mechanism.
Rather, we augment the $|K|$ output states with (approximately) $|K|$
``faulty states'', each one of them forms a short detour over the cyclic
structure of $K$;
refer to Figure~\ref{figure:AU:turn-chart} for the state diagram of \AlgAU{},
where the output states and the faulty states are marked by integers with
(wide) bars and hats, respectively.
Upon detecting a clock discrepancy, a node residing in an output state $s$
moves to the faulty state associated with $s$ and stays there until certain
conditions are satisfied and the node may complete the faulty detour and
return to a nearby output state (though, not to the original state $s$).
This mechanism is designed so that clock discrepancies are resolved in a
gradual ``closing the gap'' fashion.

The conditions that determine when a faulty node may return to an output state
and the conditions for moving to a faulty state when sensing a faulty neighbor
without being directly involved in a clock discrepancy are the key to the
stabilization guarantees of \AlgAU{}.
In particular, the algorithm takes a relatively cautious approach for switching
between output and faulty states, that, as it turns out, allows us to avoid
``vicious cycles'' and ultimately bound the stabilization time as a function
of
$|K| = O (D)$.

\subsection{Constructing the Self-Stabilizing Algorithm}
\label{section:AU:construction}
The design of \AlgAU{} relies upon the following definitions.

\begin{definition*}%
[turns, able, faulty]
Fix
$k = 3 D + 2$.
The states of \AlgAU{}, referred to hereafter as \emph{turns}, are partitioned
into a set
$\Able{T}
=
\{ \Able{\ell} \mid \ell \in \Integers, 1 \leq |\ell| \leq k \}$
of \emph{able} turns and a set
$\Faulty{T}
=
\{ \Faulty{\ell} \mid \ell \in \Integers, 2 \leq |\ell| \leq k \}$
of \emph{faulty} turns.
A node residing in an able (resp. faulty) turn is said to be \emph{able}
(resp., \emph{faulty}).
\end{definition*}

\begin{definition*}%
[levels]
Throughout \Sect{}~\ref{section:AU}, we refer to the integers
$\ell \in \Integers$,
$1 \leq |\ell| \leq k$,
as \emph{levels} and define the level of turn
$\Able{\ell} \in \Able{T}$
(resp.,
$\Faulty{\ell} \in \Faulty{T}$)
to be $\ell$.
We denote the level of (the turn of) a node
$v \in V$
at time
$t \in \Integers_{\geq 0}$
by $\lambda_{v}^{t}$ and the set of levels sensed by $v$ at time $t$ by
$\Lambda_{v}^{t} = \{ \lambda_{u}^{t} \mid u \in \Neighbors^{+}(v) \}$.
For a level $\ell$, let
$\Level^{t}(\ell) = \{ v \in V \mid \lambda_{v}^{t} = \ell \}$
be the set of nodes whose level at time $t$ is $\ell$.
This notation is extended to level subsets $B$, defining
$\Level^{t}(B) = \bigcup_{\ell \in B} \Level^{t}(\ell)$.
\end{definition*}

\begin{definition*}%
[forward operator, adjacent]
For a level $\ell$, let
\[
\forward(\ell)
\, = \,
\begin{cases}
1, & \ell = -1 \\
-k, & \ell = k \\
\ell + 1, & \text{otherwise}
\end{cases} \, .
\]
Based on that, we define the \emph{forward operator} $\forward^{j}(\ell)$,
$j = 1, 2, \dots$,
by setting
$\forward^{1}(\ell) = \forward(\ell)$
and
$\forward^{j + 1}(\ell) = \forward(\forward^{j}(\ell))$.
Observing that the forward operator is bijective for each $j$, we extend it
to negative superscripts by setting
$\forward^{-j}(\ell) = \ell'$
if and only if
$\forward^{+j}(\ell') = \ell$.
Levels $\ell$ and $\ell'$ are said to be \emph{adjacent} if either \\
(1)
$\ell = \ell'$; \\
(2)
$\ell = \forward^{+1}(\ell')$;
or \\
(3)
$\ell = \forward^{-1}(\ell')$.
\end{definition*}

\begin{definition*}%
[outwards operator, outwards, inwards]
Given a level $\ell$ and an integer parameter
$-|\ell| < j \leq k - |\ell|$,
the \emph{outwards operator} $\outwards^{j}(\ell)$ returns the unique level
$\ell'$ that satisfies
(1)
$\Sign(\ell') = \Sign(\ell)$;
and
(2)
$|\ell'| = |\ell| + j$.
This means in particular that if $j$ is positive, then
$|\ell'| > |\ell|$,
and if $j$ is negative, then
$|\ell'| < |\ell|$.
If
$\ell' = \outwards^{j}(\ell)$
for a positive (resp., negative) $j$, then we refer to level $\ell'$ as being
$|j|$ units \emph{outwards} (resp., \emph{inwards}) of $\ell$.

Let
$\Outwards^{>}(\ell)
=
\left\{ \outwards^{j}(\ell) \mid 0 < j \leq k - |\ell| \right\}$
and let
$\Outwards^{\geq}(\ell) = \Outwards^{>}(\ell) \cup \{ \ell \}$
and
$\Outwards^{\gg}(\ell) = \Outwards^{>}(\ell) - \{ \outwards^{+1}(\ell) \}$.
Likewise, let
$\Outwards^{<}(\ell)
=
\left\{ \outwards^{j}(\ell) \mid -|\ell| < j < 0 \right\}$
and let
$\Outwards^{\leq}(\ell) = \Outwards^{<}(\ell) \cup \{ \ell \}$
and
$\Outwards^{\ll}(\ell) = \Outwards^{<}(\ell) - \{ \outwards^{-1}(\ell) \}$.
\end{definition*}

\begin{definition*}%
[protected, good]
An edge
$e = (u, v) \in E$
is said to be \emph{protected} at time
$t \in \Integers_{\geq 0}$
if levels $\lambda_{v}^{t}$ and
$\lambda_{u}^{t}$ are adjacent.
A node
$v \in V$
is said to be \emph{protected} at time $t$ if all its incident edges are
protected.
Let
$\ProtectedV^{t} \subseteq V$
and
$\ProtectedE^{t} \subseteq E$
denote the set of nodes and edges, respectively, that are protected at time
$t$.
A protected node that does not sense any faulty turn is said to be
\emph{good}.
The graph $G$ is said to be \emph{protected} (resp., \emph{good}) at time $t$
if all its nodes are protected (resp., good).
\end{definition*}

We are now ready to complete the description of \AlgAU{}.
The
$2 k$
levels are identified with the AU clock values, associating $\phi^{+j}(\cdot)$
and $\phi^{-j}(\cdot)$ with the $+j$ and $-j$ operations, respectively, of the
corresponding cyclic group.
Moreover, we identify the output state set of \AlgAU{} with the set $\Able{T}$
of able turns and regard the faulty turns as the remaining (non-output)
states.

For the state transition function of \AlgAU{}, consider a node
$v \in V$
residing in a turn
$\nu \in \Able{T} \cup \Faulty{T}$
at time $t$ and suppose that $v$ is activated at time $t$.
Node $v$ remains in turn $\nu$ during step $t$ unless certain conditions on
$\nu$ are satisfied, in which case, $v$ performs a state transition that
belongs to one of the following three types (refer to
Table~\ref{table:AU:state-transition} for a summary and to
Figure~\ref{figure:AU:turn-chart} for an illustration):
\begin{itemize}

\item
Suppose that $v$'s turn at time $t$ is
$\nu = \Able{\ell} \in \Able{T}$,
$1 \leq |\ell| \leq k$.
Node $v$ performs a \emph{type able-able (AA)} transition in step $t$ and
updates its turn to
$\Able{\ell'} \in \Able{T}$,
where
$\ell' = \forward^{+1}(\ell)$,
if and only if
(1)
$v$ is good at time $t$;
and
(2)
$\Lambda_{v}^{t} \subseteq \{ \ell, \ell' \}$.

\item
Suppose that $v$'s turn at time $t$ is
$\nu = \Able{\ell} \in \Able{T}$,
$2 \leq |\ell| \leq k$.
Node $v$ performs a \emph{type able-faulty (AF)} transition in step $t$ and
updates its turn to
$\Faulty{\ell} \in \Faulty{T}$
if and only if at least one of the following two conditions is satisfied:
(1)
$v$ is not protected at time $t$;
or
(2)
$v$ senses turn $\Faulty{\ell'}$ at time $t$, where
$\ell' = \outwards^{-1}(\ell)$.

\item
Suppose that $v$'s turn at time $t$ is
$\nu = \Faulty{\ell} \in \Faulty{T}$,
$2 \leq |\ell| \leq k$.
Node $v$ performs a \emph{type faulty-able (FA)} transition in step $t$ and
updates its turn to
$\Able{\ell'} \in \Able{T}$,
where
$\ell' = \outwards^{-1}(\ell)$
is the level one unit inwards of $\ell$, if and only if $v$ does not sense any
level in
$\Outwards^{>}(\ell)$.

\end{itemize}

\subsection{Correctness and Stabilization Time Analysis}
\label{section:AU:analysis}
In this section, we establish the correctness and stabilization time
guarantees of \AlgAU{}.
First, in \Sect{}~\ref{section:AU:fundamental-properties}, we
present (and prove) certain fundamental invariants and general observations
regarding the operation of \AlgAU{}.
This allows us to prove in
\Sect{}~\ref{section:AU:post-stabilization-dynamics} that in
the context of \AlgAU{}, stabilization corresponds to reaching a good graph.
Following that, we focus on proving that the graph is guaranteed to become
good by time
$O (R(D^{3}))$.
This is done in three stages, presented in \Sect{}
\ref{section:AU:towards-out-protected},
\ref{section:AU:out-protected-to-justified}, and
\ref{section:AU:justified-to-good}.

\subsubsection{Fundamental Properties.}
\label{section:AU:fundamental-properties}
The following additional two definitions play a central role in the analysis
of \AlgAU{}.

\begin{definition*}%
[out-protected, $\ell$-out-protected]
We say that a node
$v \in V$
of level $\ell$ is \emph{out-protected} at time
$t \in \Integers_{\geq 0}$
if
$\Lambda_{v}^{t} \cap \Outwards^{\gg}(\lambda_{v}^{t}) = \emptyset$.
In other words, $v$ is out-protected at time $t$ if any edge
$(u, v) \in E - \ProtectedE^{t}$
satisfies either
(1)
$\Sign(\lambda_{u}^{t}) \neq \Sign(\lambda_{v}^{t})$;
or
(2)
$\lambda_{u}^{t} \in \Outwards^{\ll}(\lambda_{v}^{t})$.
Notice that the nodes in level
$\ell \in \{ -k, -k + 1, k - 1, k \}$
are always (vacuously) out-protected.
Let
$\OutProtectedV^{t} \subseteq V$
denote the set of nodes that are out-protected at time
$t \in \Integers_{\geq 0}$.

The graph $G$ is said to be \emph{out-protected} at time
$t \in \Integers_{\geq 0}$
if
$V = \OutProtectedV^{t}$.
Given a level $\ell$, the graph is said to be \emph{$\ell$-out-protected} at
time $t$ if 
$\Level^{t} \left( \Outwards^{\geq}(\ell) \right)
\subseteq
\OutProtectedV^{t}$.
Notice that the graph is out-protected if and only if it is both
$1$-out-protected and $(-1)$-out-protected, which means that if edge
$(u, v) \notin \ProtectedE^{t}$,
then
$\Sign(\lambda_{u}^{t}) \neq \Sign(\lambda_{v}^{t})$.
\end{definition*}

\begin{definition*}%
[distance]
The \emph{distance} between levels $\ell$ and $\ell'$, denoted by
$\Distance(\ell, \ell')$,
is defined by the recurrence
\[
\Distance(\ell, \ell')
\, = \,
\begin{cases}
0, & \ell = \ell' \\
1 + \min \{
\Distance(\ell, \forward^{-1}(\ell')), \Distance(\ell, \forward^{+1}(\ell'))
\},
& \ell \neq \ell'
\end{cases} \, ;
\]
notice that this is indeed a distance function in the sense that it is
symmetric and obeys the triangle inequality.\footnote{%
To distinguish the level distance function from the distance function of the
graph $G$, we denote the latter by
$\Distance_{G}(\cdot, \cdot)$.}
\end{definition*}

We are now ready to state the fundamental properties of \AlgAU{}, cast in
\Obs{}
\ref{observation:protected-edge-remains}--\ref{observation:path-remains-protected}.

\begin{observation} \label{observation:protected-edge-remains}
If an edge
$e = (u, v) \in \ProtectedE^{t}$
and
$\{ \lambda^{t}_{u}, \lambda^{t}_{v} \} \neq \{ -k, k \}$,
then
$e \in \ProtectedE^{t + 1}$.
\end{observation}
\begin{proof}
Consider first the case that
$\lambda_{u}^{t} = \lambda_{v}^{t} = \ell$.
If
$\ell < 0$,
then
$\{ \lambda_{u}^{t + 1}, \lambda_{v}^{t + 1} \}
\subseteq
\{ \ell, \forward^{+1}(\ell) \}$,
thus $e$ remains protected at time
$t + 1$.
If
$\ell > 0$,
then it may be the case that the level of one of the two nodes, say $u$,
decreases in step $t$ due to a type FA transition so that
$\lambda_{u}^{t + 1} = \forward^{-1}(\ell)$.
But this means that $v$ is not good at time $t$ (it has at least one faulty
neighbor), hence it cannot experience a type AA transition, implying
that
$\lambda_{v}^{t + 1} \in \{ \ell, \forward^{-1}(\ell) \}$.
Therefore, $e$ remains protected at time
$t + 1$
also in this case.

Assume now that
$\lambda_{u}^{t} = \ell$
and
$\lambda_{v}^{t} = \forward^{+1}(\ell)$
for a level
$\ell \neq k$.
Notice that $v$ cannot experience a type AA transition in step $t$ as
$\ell = \forward^{-1}(\lambda_{v}^{t}) \in \Lambda_{v}^{t}$.
On the other hand, $u$ can experience a type FA transition only if
$\ell < 0$
which results in
$\lambda_{u}^{t + 1} = \forward^{+1}(\ell)$.
Therefore,
$\{ \lambda_{u}^{t + 1}, \lambda_{v}^{t + 1} \}
\subseteq
\{ \ell, \forward^{+1}(\ell) \}$
and $e$ remains protected at time
$t + 1$.
\end{proof}

\begin{observation} \label{observation:protected-node-remains}
If a node
$v \in \ProtectedV^{t}$
and
$\lambda^{t}_{v} \notin \{ -k, k \}$,
then
$v \in \ProtectedV^{t + 1}$.
\end{observation}
\begin{proof}
Follows directly from \Obs{}~\ref{observation:protected-edge-remains}.
\end{proof}

\begin{observation} \label{observation:out-protected-node-remains}
If a node
$v \in \OutProtectedV^{t}$,
then
$v \in \OutProtectedV^{t + 1}$.
\end{observation}
\begin{proof}
Follows from \Obs{}~\ref{observation:protected-edge-remains} by recalling
that
$\Level^{t}(\ell) \subseteq \OutProtectedV^{t}$
for every
$\ell \in \{ -k, -k + 1, k - 1, k \}$.
\end{proof}

\begin{observation} \label{observation:level-change-yields-out-protected}
For a node
$v \in V$,
if
$\lambda_{v}^{t + 1} \neq \lambda_{v}^{t}$,
then
$v \in \OutProtectedV^{t + 1}$.
\end{observation}
\begin{proof}
Follows from
\Obs{}~\ref{observation:out-protected-node-remains} as node $v$ cannot
change its level in step $t$ unless it is out-protected at time $t$.
\end{proof}

\begin{observation} \label{observation:levels-dont-get-apart}
If an edge
$(u, v) \in E - \ProtectedE^{t}$
with
$\lambda_{u}^{t} < \lambda_{v}^{t}$,
then
$\lambda_{u}^{t}
\leq
\lambda_{u}^{t + 1}
<
\lambda_{v}^{t + 1}
\leq
\lambda_{v}^{t}$.
\end{observation}
\begin{proof}
Follows by recalling that a node that is not protected at time $t$ cannot
experience a type AA transition in step $t$ and that it can experience a
type FA transition in step $t$ only if it does not sense any level
(strictly) outwards of its own.
\end{proof}

\begin{observation} \label{observation:ell-out-protected-remains}
If $G$ is $\ell$-out-protected at time $t$, then $G$ remains
$\ell$-out-protected at time
$t + 1$.
\end{observation}
\begin{proof}
Follows from \Obs{} \ref{observation:out-protected-node-remains} and
\ref{observation:level-change-yields-out-protected}.
\end{proof}

\begin{observation} \label{observation:bound-level-distance-by-path-length}
Consider a path $P$ of length $d$ between nodes
$u \in V$
and
$v \in V$
in $G$.
If
$E(P) \subseteq \ProtectedE^{t}$,
then
$\Distance(\lambda_{u}^{t}, \lambda_{v}^{t}) \leq d$.
\end{observation}
\begin{proof}
By induction on $d$.
The assertion clearly holds if
$d = 0$
which implies that
$u = v$.
Consider a
$(u, v)$-path
$P$ of length
$d > 0$
and let $v'$ be the node that precedes $v$ in $P$.
By applying the inductive hypothesis to the
$(u, v')$-prefix
of $P$, we conclude that
$\Distance(\lambda_{u}^{t}, \lambda_{v'}^{t}) \leq d - 1$.
As
$(v', v) \in \ProtectedE^{t}$,
we conclude that
$\Distance(\lambda_{u}^{t}, \lambda_{v}^{t})
\leq
\Distance(\lambda_{u}^{t}, \lambda_{v'}^{t}) + 1
\leq
d$,
thus establishing the assertion.
\end{proof}

\begin{observation} \label{observation:contiguous-levels}
If
$\ProtectedV^{t} = V$,
then there exists a level $\ell$ and an integer
$0 \leq d \leq D$
such that
$V
=
\Level^{t} \left(
\left\{ \forward^{+j}(\ell) \mid 0 \leq j \leq d \right\}
\right)$.
\end{observation}
\begin{proof}
Follows by applying
\Obs{}~\ref{observation:bound-level-distance-by-path-length} to the
shortest paths in the graph $G$ whose lengths are at most $D$.
\end{proof}

\begin{observation} \label{observation:path-remains-protected}
Consider a path $P$ of length $d$ emerging from a node
$v \in V$
and assume that
$E(P) \subseteq \ProtectedE^{t}$
(resp.,
$V(P) \subseteq \ProtectedV^{t}$).
Fix some time
$t' \geq t$
and assume that
$|\lambda_{v}^{s}| < k - d$
for every
$t \leq s \leq t'$.
Then,
$E(P) \subseteq \ProtectedE^{t'}$
(resp.,
$V(P) \subseteq \ProtectedV^{t'}$).
\end{observation}
\begin{proof}
Fix a time $s$.
\Obs{}~\ref{observation:bound-level-distance-by-path-length} ensures that
if
$E(P) \subseteq \ProtectedE^{s}$
(resp.,
$V(P) \subseteq \ProtectedV^{s}$)
and
$|\lambda_{v}^{s}| < k - d$,
then
$|\lambda_{u}^{s}| < k$
for every node $u$ in $P$.
This implies that
$E(P) \subseteq \ProtectedE^{s + 1}$
(resp.,
$V(P) \subseteq \ProtectedV^{s + 1}$)
due to \Obs{}~\ref{observation:protected-edge-remains} (resp.,
\Obs{}~\ref{observation:protected-node-remains}).
The assertion is now established by induction on
$s = t, t + 1, \dots, t' - 1$.
\end{proof}

\subsubsection{Post-Stabilization Dynamics.}
\label{section:AU:post-stabilization-dynamics}
In this section, we show that the stabilization of \AlgAU{} is reduced to
reaching a good graph.
This is stated formally in the following two lemmas.

\begin{lemma} \label{lemma:all-good-remains}
If $G$ is good at time $t$, then $G$ remains good at time
$t + 1$.
\end{lemma}
\begin{proof}
If all nodes are good at time $t$, then the only possible state transitions in
step $t$ are of type AA.
Observing that an edge
$(u, v)$
with
$\lambda_{u}^{t} = k$
and
$\lambda_{v}^{t} = -k$
does not become non-protected via type AA transitions, we conclude by
\Obs{}~\ref{observation:protected-edge-remains} that
$\ProtectedE^{t + 1} = E$
and hence,
$\ProtectedV^{t + 1} = V$.
Since a type AA transition does not change the turn of a node from able
to faulty, it follows that all nodes remain able at time
$t + 1$,
hence all nodes are good at time
$t + 1$.
\end{proof}

\begin{lemma} \label{lemma:all-good-ensures-fast-progress}
Assume that $G$ is good at time $t$.
For
$i = 0, 1, \dots$,
each node
$v \in V$
experiences at least $i$ type AA transitions during the time
interval
$\left[ t, \varrho^{D + i}(t) \right)$.
\end{lemma}
\begin{proof}
\Lem{}~\ref{lemma:all-good-remains} ensures that all nodes remain
good, and in particular protected, from time $t$ onwards.
For
$i = 0, 1, \dots$,
let
$\tau(i) = \varrho^{i}(t)$
and let $\ell_{\min}(i)$ and $d(i)$ be the level $\ell$ and integer $d$
promised in \Obs{}~\ref{observation:contiguous-levels} when applied to
time $\tau(i)$.
Since all nodes are good throughout the time interval
$I = [\tau(i), \tau(i + 1))$,
it follows that every node
$v \in \Level^{\tau(i)} \left( \ell_{\min}(i) \right)$
experiences at least one type AA transition during $I$ (in particular,
$v$ experiences a type AA transition upon its first activation during
$I$), hence
$\ell_{\min}(i + 1) > \ell_{\min}(i)$.
The assertion follows by \Obs{}~\ref{observation:contiguous-levels}
ensuring that
$d(0) \leq D$.
\end{proof}

\subsubsection{Towards an Out-Protected Graph.}
\label{section:AU:towards-out-protected}
Our goal in in the remainder of \Sect{}~\ref{section:AU:analysis} is to
establish an upper bound on the time it takes until the graph becomes good.
In the current section, we make the first step towards achieving this goal by
bounding the time it takes for the graph to become out-protected, starting
with the following lemma.

\begin{lemma} \label{lemma:bound-FA-transition-time}
Assume that $G$ is $\ell$-out-protected,
$2 \leq |\ell| \leq k$,
at time $t$.
If the turn of a node
$v \in V$
at time $t$ is $\Faulty{\ell}$, then $v$ experiences a type FA transition
before time
$\varrho^{2 (k - |\ell|) + 1}(t)$.
\end{lemma}
\begin{proof}
\Obs{}~\ref{observation:out-protected-node-remains} ensures that
$v \in \OutProtectedV^{t'}$
for every
$t' \geq t$.
For
$i = 0, 1, \dots$,
let
$\tau(i) = \varrho^{i}(t)$.
We prove by induction on
$k - |\ell|$
that $v$ experiences a type FA transition before time
$\tau(2 (k - |\ell|) + 1)$,
thus establishing the assertion.
For the induction's base, notice that if the turn of node $v$ at time $t$ is
$\Faulty{k}$ (resp., $\Faulty{-k}$), then $v$ is guaranteed to experience a
type FA transition, moving to state
$\Able{k - 1}$
(resp.,
$\Able{-k + 1}$),
upon its next activation and in particular before time
$\varrho(t) = \tau(1)$.

Assume that
$2 \leq |\ell| \leq k - 1$.
If $v$ is in turn $\Faulty{\ell}$ when a neighbor $u$ of $v$ in turn
$\Able{\outwards^{+1}(\ell)}$
is activated, then $u$ experiences a type AF transition, moving to state
$\Faulty{\outwards^{+1}(\ell)}$.
Moreover, as long as $v$ is faulty, no neighbor of $v$ can move from level
$\ell$ to level $\outwards^{+1}(\ell)$.
Since $v$ has no neighbors in levels belonging to $\Outwards^{\gg}(\ell)$
(recall that $v$ is out-protected), it follows that as long as $v$ does not
experience a type FA transition, no neighbor of $v$ can move to
level $\outwards^{+1}(\ell)$ from another level and thus, no neighbor of $v$
can move to turn $\Able{\outwards^{+1}(\ell)}$ from another turn.
Therefore, it is guaranteed that at time $\tau(1)$, all neighbors $u$ of $v$
whose level satisfies
$\lambda_{u}^{\tau(1)} = \outwards^{+1}(\ell)$
are faulty.
By the inductive hypothesis, these nodes $u$ experience a type FA transition,
moving to turn $\Able{\ell}$, before time
$\tau(1 + 2 (k - |\ell| - 1) + 1) = \tau(2 (k - |\ell|))$.
In the subsequent activation of $v$, which occurs before time
$\varrho(\tau(2 (k - |\ell|))) = \tau(2 (k - |\ell|) + 1)$,
$v$ experiences a type FA transition, thus establishing the assertion.
\end{proof}

\Lem{}~\ref{lemma:bound-FA-transition-time} is the main ingredient in proving
the following key lemma.

\begin{lemma} \label{lemma:levels-get-closer}
Consider an edge
$(u, v) \in E - \ProtectedE^{t}$
with
$\lambda_{u}^{t} < \lambda_{v}^{t}$.
If $G$ is $\ell$-out-protected at time $t$ for
$\ell \in \{ \lambda_{u}^{t}, \lambda_{v}^{t} \}$,
then there exists a time
$t < t^{*} \leq \varrho^{2 (k - |\ell|) + 2}(t_{0})$
such that \\
(1)
$\lambda_{u}^{t^{*}} \geq \lambda_{u}^{t}$; \\
(2)
$\lambda_{v}^{t^{*}} \leq \lambda_{v}^{t}$;
and \\
(3)
at least one of the inequalities in (1) and (2) is strict.
\end{lemma}
\begin{proof}
By \Obs{}~\ref{observation:levels-dont-get-apart}, it is sufficient to
prove that at least one of the two nodes $u$ and $v$ changes its level before
time
$\varrho^{2 (k - |\ell|) + 2}(t)$.
Assume that the graph is $\ell$-out-protected at time $t$ for
$\ell = \lambda_{v}^{t}$;
the proof for the case that
$\ell = \lambda_{u}^{t}$
is analogous.
Let
$t \leq t_{0} < \varrho(t)$
be the first time following $t$ at which $v$ is activated and based on that,
define the time
$t \leq t_{1} \leq \varrho(t)$
as follows:
if $v$ is in turn $\Faulty{\ell}$ at time $t$, then set
$t_{1} = t$;
otherwise ($v$ is in turn $\Able{\ell}$ at time $t$), set
$t_{1} = t_{0} + 1$
and notice that $v$ experiences a type AF transition in step $t_{0}$ (due to
the non-protected edge
$(v, v')$)
unless $v'$ changes its level beforehand.
In both cases, we know that $v$ is in turn $\Faulty{\ell}$ at time $t_{1}$.
Since \Obs{}~\ref{observation:ell-out-protected-remains} guarantees that
the graph is $\ell$-out-protected at time $t_{1}$, we can apply
\Lem{}~\ref{lemma:bound-FA-transition-time} to $v$, concluding that $v$
experiences a type FA transition, and in particular changes its level, before
time
$\varrho^{2 (k - |\ell|) + 1}(t_{1})
\leq
\varrho^{2 (k - |\ell|) + 2}(t)$,
thus establishing the assertion.
\end{proof}

Building on \Lem{}~\ref{lemma:levels-get-closer}, we can now bound the time it
takes for the graph to become $\ell$-out-protected after it is already
$\outwards^{+1}(\ell)$-out-protected.

\begin{lemma} \label{lemma:ell-out-protected-advances}
Fix a level
$1 \leq |\ell| \leq k - 1$
and assume that $G$ is $\outwards^{+1}(\ell)$-out-protected at time $t$.
Then, $G$ is $\ell$-out-protected at time
$\varrho^{(k - |\ell|) (k - |\ell| - 1)}(t)$.
\end{lemma}
\begin{proof}
For
$i = 0, 1, \dots$,
let
$\tau(i) = \varrho^{i}(t)$
and fix
$t^{*} = \tau((k - |\ell|) (k - |\ell| - 1))$.
By \Obs{}~\ref{observation:level-change-yields-out-protected}, it
suffices to prove that
$\bigcap_{t \leq t' \leq t^{*}} \Level^{t'}(\ell)
\subseteq
\OutProtectedV^{t^{*}}$.
To this end, consider a node
$v \in \bigcap_{t \leq t' \leq t^{*}} \Level^{t'}(\ell)$
and notice that by \Obs{}~\ref{observation:out-protected-node-remains},
if $v$ is out-protected at any time
$t \leq t' \leq t^{*}$,
then it remains out-protected subsequently and in particular at time $t^{*}$.
Moreover, \Obs{}~\ref{observation:protected-edge-remains} ensures that
any neighbor of $v$ whose level at time $t$ belongs to
$\Outwards^{\leq}(\ell) \cup \{ \outwards^{+1}(\ell) \}$
cannot move to a level in $\Outwards^{\gg}(\ell)$ as long as $v$ is in level
$\ell$.

So, it remains to consider a neighbor
$u \in \Neighbors(v)$
of $v$ with
$\lambda_{u}^{t} \in \Outwards^{\gg}(\ell)$
and show that the level of $u$ moves inwards and becomes adjacent to $\ell$ by
time $t^{*}$;
indeed, \Obs{}~\ref{observation:protected-edge-remains} ensures that once
$u$ reaches a level adjacent to $\ell$, it cannot move back to a level in
$\Outwards^{\gg}(\ell)$ unless $v$ leaves level $\ell$.
To this end, we define
\[\textstyle
f(\ell^{*})
\, = \,
\sum_{j = |\ell| + 2}^{|\ell^{*}|} (2 (k - j) + 2)
\]
and prove that if
$\lambda_{u}^{t} = \ell^{*} \in \Outwards^{\gg}(\ell)$,
then $u$ reaches level $\outwards^{+1}(\ell)$ by time $\tau(f(\ell^{*}))$.
The assertion is established by observing that
$f(k) = (k - |\ell|) (k - |\ell| - 1)$.

Since the graph $G$ is $\outwards^{+1}(\ell)$-out-protected at time $t$,
\Obs{}~\ref{observation:ell-out-protected-remains} guarantees that $G$ is
$\outwards^{+1}(\ell)$-out-protected at all times subsequent to $t$ and hence,
also $\ell'$-out-protected for every
$\ell' \in \Outwards^{>}(\ell)$.
Therefore, we can repeatedly apply \Lem{}~\ref{lemma:levels-get-closer} to edge
$(u, v)$
and conclude by induction on $\ell'$ that $u$ moves from level
$\ell' \in \Outwards^{\gg}(\ell)$,
$|\ell'| \leq \ell^{*}$,
to level $\outwards^{-1}(\ell')$ by time
\[\textstyle
\tau \left( \sum_{j = |\ell'|}^{|\ell^{*}|} (2 (k - j) + 2) \right) \, .
\]
The proof is then completed by plugging
$\ell' = \outwards^{+2}(\ell)$.
\end{proof}

Since the graph $G$ is $\ell$-out-protected for
$\ell \in \{ -k, -k + 1, k - 1, k \}$
already at time $0$ and since being $1$-out-protected and ($-1$)-out-protected
implies that $G$ is out-protected,
\Lem{}~\ref{lemma:ell-out-protected-advances} yields the following corollary.

\begin{corollary} \label{corollary:graph-becomes-out-protected}
There exists a time
$T_{0} \leq R(O (k^{3}))$
such that $G$ is out-protected at all times
$t \geq T_{0}$.
\end{corollary}

\subsubsection{From an Out-Protected to a Justified Graph.}
\label{section:AU:out-protected-to-justified}
In what follows, we take $T_{0}$ to be the time promised in
Corollary~\ref{corollary:graph-becomes-out-protected} and consider the
execution from time $T_{0}$ onwards.

\begin{definition*}%
[justifiably faulty, unjustifiably faulty, justified]
A node
$v \in V$
whose turn at time $t$ is $\Faulty{\ell}$,
$2 \leq |\ell| \leq k$,
is said to be \emph{justifiably faulty} if either
(1)
$v \notin \ProtectedV^{t}$;
or
(2)
$v$ admits a neighbor whose turn at time $t$ is
$\Faulty{\outwards^{-1}(\ell)}$.
A faulty node that is not justifiably faulty is said to be \emph{unjustifiably
faulty}.
We say that the graph $G$ is \emph{justified} if it does not admit any
unjustifiably faulty node.
\end{definition*}

A key feature of \AlgAU{} is that nodes do not become unjustifiably faulty
once the graph is out-protected.

\begin{lemma} \label{lemma:no-fresh-unjustifiably-faulty}
If a node
$v \in V$
is not unjustifiably faulty at time
$t \geq T_{0}$,
then $v$ is not unjustifiably faulty
at time
$t + 1$.
\end{lemma}
\begin{proof}
Assume that node $v$ is either
(1)
able at time $t$ and experiences a type AF transition in step $t$;
or
(2)
justifiably faulty at time $t$ (and remains faulty at time
$t + 1$).
In both cases, we know that $v$ admits a neighboring node
$u \in \Neighbors(v)$
that satisfies at least one of the following two conditions:
(i)
$\lambda_{u}^{t}$ is not adjacent to $\lambda_{v}^{t}$;
or
(ii)
$\lambda_{u}^{t} = \outwards^{-1}(\lambda_{v}^{t})$
and $u$ is faulty at time $t$.

Assuming that condition (i) holds, we know that
$\Sign(\lambda_{u}^{t}) \neq \Sign(\lambda_{v}^{t})$
as $G$ is out-protected at time $t$.
Since $v$ is faulty at time
$t + 1$,
it follows that
$\lambda_{v}^{t + 1} = \lambda_{v}^{t}$
with
$|\lambda_{v}^{t + 1}| \geq 2$.
Thus,
$\Sign(\lambda_{u}^{t + 1}) = \Sign(\lambda_{u}^{t})$
and edge
$(u, v)$
remains non-protected at time
$t + 1$.
Assuming that condition (ii) holds, node $u$ cannot experience a type FA
transition in step $t$ as
$\lambda_{v}^{t} = \outwards^{+1}(\lambda_{u}^{t}) \in \Lambda_{u}^{t}$,
thus it remains faulty at time
$t + 1$.
Therefore, we conclude that $v$ is justifiably faulty at time
$t + 1$.
\end{proof}

\sloppy
Corollary~\ref{corollary:graph-becomes-justified} is now derived
by combining Corollary~\ref{corollary:graph-becomes-out-protected} and
\Lem{}~\ref{lemma:no-fresh-unjustifiably-faulty}, recalling that
\Lem{}~\ref{lemma:bound-FA-transition-time} ensures that if the graph is
out-protected at time $T_{0}$, then any (justifiably or) unjustifiably faulty
node experiences a type FA transition, and in particular stops being
unjustifiably faulty, before time
$\varrho^{O (k)}(T_{0}) \leq R(O(k^{3}))$.
\par\fussy

\begin{corollary} \label{corollary:graph-becomes-justified}
There exists a time
$T_{0} \leq T_{1} \leq R(O(k^{3}))$
such that $G$ is justified at all times
$t \geq T_{1}$.
\end{corollary}

\subsubsection{From a Justified to a Good Graph.}
\label{section:AU:justified-to-good}
In what follows, we take $T_{1}$ to be the time promised in
Corollary~\ref{corollary:graph-becomes-justified} and consider the execution
from time $T_{1}$ onwards.
In the current section, we complete the analysis by up-bounding the time it
takes for the graph to become good following time $T_{1}$, starting with the
following lemma.

\begin{lemma} \label{lemma:protected-and-justified-implies-good}
If $G$ is protected at time
$t \geq T_{1}$,
then $G$ is good at time $t$.
\end{lemma}
\begin{proof}
Assume by contradiction that the graph admits faulty nodes at time $t$ and
among these nodes, let
$v \in V$
be a node that minimizes
$|\lambda_{v}^{t}|$.
Since
$t \geq T_{1}$,
Corollary~\ref{corollary:graph-becomes-justified} ensures that $v$ is
justifiably faulty at time $t$.
The assumption that $G$ is protected implies that $v$ admits a neighbor
$u \in \Neighbors(v)$
whose turn at time $t$ is
$\Faulty{\outwards^{-1}(\lambda_{v}^{t})}$,
in contradiction to the choice of $v$.
\end{proof}

Owing to \Lem{}~\ref{lemma:protected-and-justified-implies-good},
our goal in the remainder of this section is to prove that it does not take
too long after time $T_{1}$ for the graph to become protected.
\Lem{}~\ref{lemma:non-protected-node-becomes-protected} plays a key role in
in achieving this goal.

\begin{lemma} \label{lemma:non-protected-node-becomes-protected}
If a node
$v \in V - \ProtectedV^{t}$
for some time
$t \geq T_{1}$,
then there exists a time
$t \leq t' \leq \varrho^{k (k - 1)}(t)$
such that
$v \in \ProtectedV^{t'}$
with
$\lambda_{v}^{t'} \in \{ -1, 1 \}$.
\end{lemma}
\begin{proof}
Since the graph is out-protected at all times after
$T_{1} \geq T_{0}$,
it follows that if edge
$(v, v') \in E - \ProtectedE^{t}$,
then
(1)
$\Sign(\lambda_{v}^{t}) \neq \Sign(\lambda_{v'}^{t})$;
and
(2)
$\Distance(\lambda_{v}^{t}, \lambda_{v'}^{t}) \geq 2$.
\Obs{}~\ref{observation:levels-dont-get-apart} and
\Lem{}~\ref{lemma:levels-get-closer} guarantee that the levels of $v$ and $v'$
move inwards until they meet with
$\{ \lambda_{v}^{t'}, \lambda_{v'}^{t'} \} = \{ -1, 1 \}$
at some time
$t \leq t' \leq \varrho^{z}(t)$
for
$z = \sum_{j = 2}^{k} 2 (k - j) + 2 = k (k - 1)$.
The assertion follows as this is true for all edges
$(v, v') \in E - \ProtectedE^{t}$.
\end{proof}

\Lem{}~\ref{lemma:non-protected-node-becomes-protected} by itself does not
complete the analysis as it does not address protected nodes that become
non-protected (alas, still out-protected).
The following lemma provides a sufficient condition for the whole graph to
become protected.

\begin{lemma} \label{lemma:level-progression-implies-protected-graph}
Consider a node
$v \in V$
and assume that there exist times
$T_{1} \leq t < t'$
such that
(i)
$\lambda_{v}^{t} = 1$;
and
(ii)
$\lambda_{v}^{t'} = 2 D + 2$.
Then $G$ is protected at time $t'$.
\end{lemma}
\begin{proof}
By \Lem{} \ref{lemma:all-good-remains} and
\ref{lemma:protected-and-justified-implies-good}, if all nodes are protected
at some time after time $T_{1}$, then all nodes remain (good and
hence) protected indefinitely.
Therefore, we establish the assertion by proving the following claim and
plugging
$d = D$:
Assume that there exist levels
$1 \leq \ell < \ell' \leq 2 D + 2$
with
$\ell' - \ell = 2 d + 1$
such that \\
(I)
$v$ moves in step $t$ from level
$\lambda_{v}^{t} = \ell$
to level
$\lambda_{v}^{t + 1} = \ell + 1$; \\
(II)
$v$ moves in step
$t' - 1$
from level
$\lambda_{v}^{t' - 1} = \ell' - 1$
to level
$\lambda_{v}^{t'} = \ell'$;
and \\
(III)
$\ell < \lambda_{v}^{s} < \ell'$
for all
$t < s < t'$. \\
Then all nodes at distance at most $d$ from $v$ are protected at time $t'$.

Node $v$ can move from level $\ell$ to level
$\ell + 1$
in step $t$ only if it experiences a type AA transition, which requires $v$ to
be protected at time $t$.
By \Obs{}~\ref{observation:protected-node-remains}, $v$ remains protected
throughout the time interval
$[t, t']$.

We prove that all other nodes in
$B(v, d) = \{ u \in V \mid \Distance_{G}(u, v) \leq d \}$
are protected at time $t'$ by induction on $d$.
The assertion holds trivially for
$d = 0$
as
$B(v, 0) = \{ v \}$.
Assume that the assertion holds for
$d - 1 \geq 0$
and consider a node
$u \in B(v, d)$.
Let $P$ be a shortest
$(v, u)$-path
in $G$ and let $w$ be the node succeeding $v$ along $P$.

Since $v$ experiences
$2 d + 1$
type AA transitions while moving from level $\ell$ to level $\ell'$ during the
time interval
$[t, t']$,
there must exist times
$t < t_{w} \leq t'_{w} < t'$
such that \\
(I)
$w$ moves in step $t_{w}$ from level
$\lambda_{w}^{t_{w}} = \ell + 1$
to level
$\lambda_{w}^{t_{w} + 1} = \ell + 2$; \\
(II)
$w$ moves in step
$t'_{w} - 1$
from level
$\lambda_{w}^{t'_{w} - 1} = \ell' - 2$
to level
$\lambda_{w}^{t'_{w}} = \ell' - 1$;
and \\
(III)
$\ell + 1 < \lambda_{w}^{s} < \ell' - 1$
for all
$t_{w} < s < t'_{w}$. \\
By the inductive hypothesis, all nodes in
$B(u, d - 1)$,
and in particular the nodes along the
$(w, u)$-suffix
of $P$, are protected at time $t'_{w}$, hence all nodes in $P$ are protected
at time
$t < t'_{w} < t'$.
Recalling that
$1 \leq \ell < \lambda_{v}^{s} \leq \ell' \leq 2 D + 2$
for all
$t'_{w} \leq s \leq t'$,
we employ \Obs{}~\ref{observation:path-remains-protected} to conclude
that all nodes in $P$ are protected at time $t'$, thus establishing the
assertion.
\end{proof}

\Lem{}~\ref{lemma:level-progression-implies-protected-graph} allows us to
establish \Lem{}~\ref{lemma:post-grounded-implies-protected} for which we need
the following additional definition.

\begin{definition*}%
[grounded]
A path $P$ of length at most $D$ in $G$ is said to be \emph{grounded} at time
$t$ if
(1) 
$V(P) \subseteq \ProtectedV^{t}$;
and
(2)
$P$ has an endpoint $u$ satisfying
$\lambda_{u}^{t} \in \{ -1, 1 \}$.
A node
$v \in V$
is said to be \emph{grounded} at time $t$ if it belongs to a grounded path.
\end{definition*}

\begin{lemma} \label{lemma:post-grounded-implies-protected}
If a node
$v \in V$
is grounded at time
$t \geq T_{1}$,
then
$v \in \ProtectedV^{t'}$
for all
$t' \geq t$.
\end{lemma}
\begin{proof}
The fact that node $v$ is grounded at time $t$ means in particular that
$v \in \ProtectedV^{t}$
so assume by contradiction that
$v \notin \ProtectedV^{t'}$
for a time
$t' > t$.
Consider the path $P$ of length at most $D$ due to which $v$ is grounded at
time $t$ and let $u$ be the endpoint of $P$ that satisfies
$\lambda_{u}^{t} \in \{ -1, 1 \}$.
Since
$V(P) \subseteq \ProtectedV^{t}$,
we can apply \Obs{}~\ref{observation:path-remains-protected} to $P$ and
$u$, concluding that there exists a time
$t < s \leq t'$
such that
$|\lambda_{u}^{s}| \geq k - D$.
Since $u$ moves from level
$\lambda_{u}^{t} \in \{ -1, 1 \}$
to level $\lambda_{u}^{s}$ satisfying
$|\lambda_{u}^{s}| \geq k - D = 2 D + 2$
during the time interval
$[t, s)$,
it follows that there exist times
$t \leq r < r' \leq s$
such that $u$ moves from level
$\lambda_{u}^{r} = 1$
up to level
$\lambda_{u}^{r'} = 2 D + 2$
during the time interval
$[r, r')$.
Employing \Lem{}~\ref{lemma:level-progression-implies-protected-graph}, we
conclude that $G$ is protected from time $r'$ onwards, which contradicts the
assumption that
$v \notin \ProtectedV^{t'}$
as
$t' \geq s \geq r'$.
\end{proof}

We are now ready to prove the following lemma that, when combined with
\Lem{} \ref{lemma:all-good-remains},
\ref{lemma:all-good-ensures-fast-progress}, and
\ref{lemma:protected-and-justified-implies-good}, establishes
\Thm{}~\ref{theorem:AU} as
$k = O (D)$.

\begin{lemma} \label{lemma:graph-becomes-protected}
There exists a time
$T_{1} \leq T_{2} \leq R(O (k^{3}))$
such that $G$ is protected at time
$T_{2}$.
\end{lemma}
\begin{proof}
Fix a node
$v \in V$.
In the context of this proof, we say that $v$ is \emph{post-grounded} at time
$t$ if $v$ was grounded at some time
$T_{1} \leq t' \leq t$.
By \Lem{}~\ref{lemma:post-grounded-implies-protected}, it suffices to prove
that $v$ becomes post-grounded by time
$R(O (k^{3}))$.
In fact, since graph $G$ is out-protected after time
$T_{1} \geq T_{0}$
and since in an out-protected graph, a non-protected node becomes protected if
and only if it becomes grounded, it follows that $G$ becomes protected exactly
when all its nodes become post-grounded.

For
$t \geq T_{1}$,
let
$G_{\mathrm{p}}^{t} = (V, \ProtectedE^{t})$.
Assuming that $G$ is still not protected at time $t$ (i.e., that
$\ProtectedV^{t} \subsetneq V$),
let $x^{t}$ be a node
$x \in V - \ProtectedV^{t}$
that minimizes
$\Distance_{G_{\mathrm{p}}^{t}}(v, x)$,
and among those, a node that minimizes $|\lambda_{x}^{t}|$ (breaking the
remaining ties in an arbitrary consistent manner).
Notice that although we cannot bound the diameter of $G_{\mathrm{p}}^{t}$, the
choice of $x^{t}$ implies that
$d^{t} = \Distance_{G_{\mathrm{p}}^{t}}(v, x^{t}) \leq D$.
Let $P^{t}$ be a
$(v, x^{t})$-path
in $G_{\mathrm{p}}^{t}$ that realizes $d^{t}$.

The choice of $x^{t}$ and $P^{t}$ ensures that
$x^{t} \notin \ProtectedV^{t}$
and that
$V(P^{t}) - \{ x^{t} \} \subseteq \ProtectedV^{t}$.
If a node
$u \in V(P^{t}) - \{ x^{t} \}$
becomes non-protected in step $t$, then
$
d^{t + 1}
\leq
\Distance_{G_{\mathrm{p}}^{t}}(v, u)
<
\Distance_{G_{\mathrm{p}}^{t}}(v, x^{t})
=
d^{t}
$.
Moreover, if $x^{t}$ remains non-protected at time
$t + 1$,
then
$d^{t + 1} \leq d^{t}$.
The more interesting case occurs when
$V(P^{t}) - \{ x^{t} \} \subseteq \ProtectedV^{t + 1}$
and $x^{t}$ also becomes protected in step $t$ which means that all nodes in
$P^{t}$ are protected at time
$t + 1$.
Recalling that the graph is out-protected after time
$T_{1} \geq T_{0}$,
we know that 
$\lambda_{x^{t}}^{t + 1} \in \{ -1, 1 \}$,
hence $P^{t}$ is grounded at time
$t + 1$
and $v$ is post-grounded from time
$t + 1$
onwards.

To complete the proof, let
$\tau(i) = \varrho^{i}(T_{1})$
for
$i = 0, 1, \dots$
and notice that \Lem{}~\ref{lemma:non-protected-node-becomes-protected}
guarantees that if $G$ is still not protected at time $\tau(i)$, then
$x^{\tau(i)}$ becomes protected before time
$\tau(i + O (k^{2}))$.
Therefore, if node $v$ is still not post-grounded at time $\tau(i)$, then
either
(1)
$v$ is post-grounded at time
$\tau(i + O (k^{2}))$;
or
(2)
$d^{\tau(i + O (k^{2}))} < d^{\tau(i)}$.
As
$0 \leq d^{t} \leq D$
for all
$t \geq T_{1}$,
we conclude that node $v$ must become post-grounded by time
$\tau(O (D \cdot k^{2})) = \tau(O (k^{3}))$.
\end{proof}

\section{Algorithms for LE and MIS}
\label{section:algorithms-LE-MIS}
In this section, we present the synchronous algorithms promised in \Thm{}
\ref{theorem:LE} and \ref{theorem:MIS}.
Specifically, our MIS algorithm, denoted by \AlgMIS{}, is developed in
\Sect{}~\ref{section:algorithm-MIS}, and our LE algorithm, denoted by
\AlgLE{}, is developed in \Sect{}~\ref{section:algorithm-LE}.

A common key ingredient in the design of \AlgMIS{} and \AlgLE{} is a
(synchronous) module denoted by \ModuleRestart{}.
This module is invoked upon detecting an illegal configuration and, as its
name implies, resets all other modules, allowing the algorithm a ``fresh
start'' from a uniform initial configuration, that is, a configuration in
which all nodes share the same \emph{initial state} $q^{*}_{0}$, chosen by the
algorithm designer.
Module \ModuleRestart{} consists of
$O (D)$
states, among them are two designated states denoted by
\ModuleRestart{}-\texttt{entry} and \ModuleRestart-\texttt{exit}:
a node enters \ModuleRestart{} by moving from a non-\ModuleRestart{} state to
\ModuleRestart{}-\texttt{entry};
a node exits \ModuleRestart{} by moving from \ModuleRestart{}-\texttt{exit} to
the initial state $q^{*}_{0}$.
The main guarantee of \ModuleRestart{} is cast in the following theorem.

\begin{theorem} \label{theorem:module-restart}
If some node is in a \ModuleRestart{} state at time $t_{0}$, then there exists a
time
$t_{0} \leq t \leq t_{0} + O (D)$
such that all nodes exit \ModuleRestart{}, concurrently, in step $t$.
\end{theorem}

A module that satisfies the promise of \Thm{}~\ref{theorem:module-restart} is
developed by Boulinier et al.~\cite{BoulinierPV2005synchronous}.
Due to some differences between the computational model used in the current
paper and the one used in \cite{BoulinierPV2005synchronous}, we provide a
standalone implementation (and analysis) of module \ModuleRestart{} in
\Sect{}~\ref{section:module-restart}, relying on algorithmic principles
similar to those used by Boulinier et al.

\subsection{Algorithm \AlgMIS{}}
\label{section:algorithm-MIS}
For clarity of the exposition, the MIS algorithm \AlgMIS{} is presented in a
procedural style;
converting it to a randomized state machine with
$O (D)$
states is straightforward.
The algorithm is designed assuming that the execution starts concurrently at
all nodes;
this assumption is plausible due to \Thm{}~\ref{theorem:module-restart} and
given the algorithm's fault detection guarantees (described in the sequel).
Throughout, we say that a node
$v \in V$
is \emph{decided} if $v$ resides in an output state;
otherwise, we say that $v$ is \emph{undecided}.
An edge is said to be \emph{decided} if at least one of its endpoints is
decided, and \emph{undecided} if both its endpoints are undecided.
Recall that in the context of the MIS problem, the output value of a decided
node $v$ is $1$ (resp., $0$) if $v$ is included in (resp., excluded from) the
constructed MIS;
we subsequently denote by $\InMIS$ (resp., $\OutMIS$) the set of decided nodes
with output $1$ (resp, $0$).

The algorithm consists of three modules, denoted by \ModuleRandPhase{},
\ModuleDetectMIS{}, and \ModuleCompete{};
all nodes participate in \ModuleRandPhase{}, whereas \ModuleDetectMIS{}
involves only the decided nodes and \ModuleCompete{} involves only the
undecided nodes.
Module \ModuleRandPhase{} runs indefinitely and divides the execution into
\emph{phases} so that for each phase $\pi$,
(1)
all nodes start (and finish) $\pi$ concurrently;
and
(2)
the length (in rounds) of $\pi$ is
$D + O (\log n)$
in expectation and whp.

The role of \ModuleDetectMIS{} is to detect \emph{local faults} among
the decided nodes, namely, two neighboring $\InMIS$ nodes or an $\OutMIS$ node
with no neighboring $\InMIS$ node.
The module runs indefinitely (over the decided nodes) and is designed so that
a local fault is detected in each round (independently) with a positive
constant probability, which means that no local fault remains undetected for
more than
$O (\log n)$
rounds whp.
Upon detecting a local fault, \ModuleDetectMIS{} invokes module
\ModuleRestart{} and the execution of \AlgMIS{} starts from scratch once
\ModuleRestart{} is exited.

Module \ModuleCompete{} is invoked from scratch in each phase, governing
the competition of the undecided nodes over the ``privilege'' to be included
in the constructed MIS.
Taking
$U \subseteq V$
to be the set of undecided nodes at the beginning of a phase $\pi$ and taking
$G(U)$ to denote the subgraph induced on $G$ by $U$, module \ModuleCompete{}
assigns (implicitly) a random variable
$Z(u) \in \Integers_{\geq 0}$
to each node
$u \in U$
so that the following three properties are satisfied: \\
(1)
$\Pr \left( \bigwedge_{w \in W} [Z(u) > Z(w)] \right)
\geq
\Omega \left( \frac{1}{|W| + 1} \right)$
for every node subset
$W \subseteq U - \{ u \}$; \\
(2)
if
$Z(u) > Z(w)$
for all nodes
$w \in \Neighbors_{G(U)}(u)$,
then $u$ joins $\InMIS$;
and \\
(3)
$u$ joins $\OutMIS$ during $\pi$ if and only if node $v$ joins $\InMIS$ for
some
$v \in \Neighbors_{G(U)}(u)$
whp.

It is well known (see, e.g.,
\cite{AlonBI1986fast, MetivierRSZ2011optimal, EmekW2013stone})
that properties (1)--(3) ensure that in expectation, a (positive) constant
fraction of the undecided edges become decided during $\pi$.
Using standard probabilistic arguments, we deduce that all edges become
decided within
$O (\log n)$
phases in expectation and whp, thus, by applying properties (1) and (2) to the
nodes of degree
$\Degree_{G(U)}(v) = 0$,
all nodes become decided within
$O (\log n)$
phases in expectation and whp.
\Thm{}~\ref{theorem:MIS} follows, again, by standard probabilistic
arguments.
We now turn to present the implementation of the three modules.

\subsubsection{Implementing Module \ModuleRandPhase{}.}
As discussed earlier, module \ModuleRandPhase{} divides the execution into
phases.
Each phase consists of a (random) prefix of length $X$ and a (deterministic)
suffix of length
$D + 2$,
where $X$ is a random variable that satisfies
(1)
$X \leq O (\log n)$
in expectation and whp;
and
(2)
$X \geq c_{0} \log n$
whp for a constant
$c_{0} > 0$
that can be made arbitrarily large.
The module is designed so that if all nodes start a phase $\pi$ concurrently,
then all nodes finish $\pi$ (and start the next phase) concurrently after
$D + 2 + X$
rounds (this guarantee holds with probability $1$).

To implement \ModuleRandPhase{}, each node
$v \in V$
maintains two variables, denoted by
$v.\varFlag \in \{ 0, 1 \}$
and
$v.\varStep \in \{ 0, 1, \dots, D + 2 \}$;
the former variable controls the length of the phase's random prefix, whereas
the latter is used to ensure that all nodes finish the phase concurrently,
exactly
$D + 2$
rounds after the random prefix is over (for all nodes).

To this end, when a phase begins, $v$ sets
$v.\varStep \gets 0$
and
$v.\varFlag \gets 1$.
As long as
$v.\varFlag = 1$,
node $v$ tosses a (biased) coin and resets
$v.\varFlag \gets 0$
with probability
$0 < p_{0} < 1$,
where
$p_{0} = p_{0}(c_{0})$
is a constant determined by $c_{0}$.
Once
$v.\varFlag = 0$,
the actions of $v$ become deterministic:
Let
$v.\varStep_{\min} = \min \{ u.\varStep : u \in \Neighbors^{+}(v) \}$.
If
$v.\varStep_{\min} < D + 2$,
then $v$ sets
$v.\varStep \gets \varStep_{\min} + 1$;
otherwise
($v.\varStep_{\min} = D + 2$),
the phase ends and a new phase begins.
On top of these rules, if, at any stage of the execution, $v$ senses a node
$u \in \Neighbors^{+}(v)$
for which
$|u.\varStep - v.\varStep| > 1$,
then $v$ invokes module \ModuleRestart{}.

To analyze \ModuleRandPhase{}, consider a phase $\pi$ that starts
concurrently for all nodes and let $X_{v}$ be the number of rounds in which
node
$v \in V$
kept
$v.\varFlag = 1$
since $\pi$ began, observing that $X_{v}$ is a
$\mathrm{Geom}(p_{0})$
random variable.
Since the random variables $X_{v}$,
$v \in V$,
are independent, we can apply \Obs{}~\ref{observation:max-geometric},
established by standard probabilistic arguments, to conclude that the random
variable
$X = \max_{v \in V} X_{v}$
satisfies
(1)
$X \leq O (\log n)$
in expectation and whp;
and
(2)
$X \geq c_{0} \log n$
whp, where the relation between $c_{0}$ and $p_{0}$ is derived from
\Obs{}~\ref{observation:max-geometric}.

\begin{observation} \label{observation:max-geometric}
Fix some constant
$0 < p \leq 1 / 2$
and let
$Y_{1}, \dots, Y_{n}$
be $n$ independent and identically distributed $\mathrm{Geom}(p)$ random
variables.
Then, the random variable
$Y = \max_{i \in [n]} Y_{i}$
satisfies
(1)
$Y \leq O (\log n)$
in expectation and whp;
and
(2)
$Y \geq c \log n$
whp for any constant
$c < \ln (2) / (2 p)$.
\end{observation}

To complete the analysis of \ModuleRandPhase{}, we introduce the following
notation and terminology.
Given a node
$v \in V$,
let $v.\varStep^{t}$ and $v.\varStep_{\min}^{t}$ denote the values of
$v.\varStep$ and $v.\varStep_{\min}$, respectively, at time $t$.
An edge
$e = \{ u, v \} \in E$
is said to be \emph{valid} at time $t$, if
$|u.\varStep^{t} - v.\varStep^{t}| \leq 1$.
Let $v_{\max}$ be a node
$v \in V$
that realizes
$X_{v} = X$.
We can now establish the following two observations.

\begin{observation} \label{observation:all-valid-implies-bounded-differences}
If all edges are valid at time $t$, then
$|u.\varStep^{t} - v.\varStep^{t}| \leq \Distance_{G}(u, v)$
for every two nodes
$u, v \in V$.
\end{observation}
\begin{proof}
Follows by a straightforward induction on
$\Distance_{G}(u, v)$.
\end{proof}

\begin{observation} \label{observation:pre-climb}
As long as
$v_{\max}.\varStep = 0$,
all edges are valid and
$v.\varStep \leq D$
for all nodes
$v \in V$.
\end{observation}
\begin{proof}
The assertion clearly holds when the phase begins and
$v.\varStep = 0$
for all nodes
$v \in V$.
\Obs{}~\ref{observation:all-valid-implies-bounded-differences} ensures
that if all edges are valid at time $t$ and
$v_{\max}.\varStep^{t} = 0$,
then
$\max_{v \in V} v.\varStep^{t} \leq D$.
The assertion follows as \ModuleRandPhase{} can invalidate a valid edge
$\{ u, v \}$
only if
$u.\varStep = v.\varStep = D + 2 > D$.
\end{proof}

Based on \Obs{} \ref{observation:all-valid-implies-bounded-differences}
and \ref{observation:pre-climb}, we can prove the following key lemma.

\begin{lemma} \label{lemma:module-rand-phase}
Suppose that node $v_{\max}$ resets
$v_{\max}.\varFlag \gets 0$
in round $t$.
Then, the following three conditions are satisfied for every
$0 \leq d \leq D$: \\
(1)
all edges are valid at time
$t + d$; \\
(2)
$v.\varStep^{t + d} \geq d$
for every node
$v \in V$;
and \\
(3)
$v.\varStep^{t + d}
\leq
\max \{ d, \Distance_{G}(v_{\max}, v) \}$
for every node
$v \in V$.
\end{lemma}
\begin{proof}
By induction on
$d = 0, 1, \dots, D$.
The base case holds by \Obs{}~\ref{observation:pre-climb} as
$v_{\max}.\varStep^{t} = 0$,
so assume that the assertion holds for
$d - 1$
and consider the situation at time
$t + d$.
The inductive hypothesis ensures that all edges are valid at time
$t - d - 1$
and that
$\max_{v \in V} v.\varStep^{t - d - 1} \leq D$,
hence all edges remain valid at time
$t + d$,
establishing condition (1).

To show that condition (2) holds, consider some node
$v \in V$.
The inductive hypothesis ensures that
$d - 1
\leq
v.\varStep_{\min}^{t + d - 1}
\leq
D$,
hence
$v.\varStep^{t + d} = v.\varStep_{\min}^{t + d - 1} + 1 \geq d$,
establishing condition (2).

For condition (3), consider some node
$v \in V$
and assume first that
$\Distance_{G}(v_{\max}, v) \leq d - 1$.
The inductive hypothesis ensures that
$v.\varStep^{t + d - 1} = d - 1$,
hence
$v.\varStep_{\min}^{t + d - 1} = d - 1$
implying that $v.\varStep$ is incremented in round
$t + d - 1$
from
$v.\varStep^{t + d - 1} = d - 1$
to
$v.\varStep^{t + d} = d = \max \{ d, \Distance_{G}(v_{\max}, v) \}$.
Now, consider the case that
$\Distance_{G}(v_{\max}, v) = d$
and let $u$ be the node that precedes $v$ along a shortest
$(v_{\max}, v)$-path in $G$.
Since
$\Distance_{G}(v_{\max}, u) = d - 1$,
we know that $u.\varStep$ is incremented in round
$t + d - 1$
from
$u.\varStep^{t + d - 1} = d - 1$
to
$u.\varStep^{t + d} = d$.
This implies that
$v.\varStep_{\min}^{t + d - 1} \leq d - 1$,
thus
$v.\varStep^{t + d} \leq d = \max \{ d, \Distance_{G}(v_{\max}, v) \}$.

\sloppy
We can now prove by a secondary induction on
$\delta = d, d + 1, \dots, D$
that
$v.\varStep^{t + d}
\leq
\Distance_{G}(v_{\max}, v)
=
\max \{ d, \Distance_{G}(v_{\max}, v) \}$
for every node
$v \in V$
with
$\Distance_{G}(v_{\max}, v) = \delta$,
thus establishing condition (3).
The base case
($\delta = d$)
of the secondary induction has already bean established, so assume that it
holds for $\delta$ and consider a node
$v \in V$
with
$\Distance_{G}(v_{\max}, v) = \delta + 1$.
Let $u$ be the node that precedes $v$ along a shortest
$(v_{\max}, v)$-path in $G$.
Since
$\Distance_{G}(v_{\max}, u) = \delta$,
we can apply the secondary inductive hypothesis, concluding that
$u.\varStep^{t + d} \leq \Distance_{G}(v_{\max}, u)$.
As edge
$\{ u, v \}$
is valid at time
$t + d$,
we conclude by
\Obs{}~\ref{observation:all-valid-implies-bounded-differences}
that
$v.\varStep^{t + d}
\leq
u.\varStep^{t + d} + 1
\leq
\Distance_{G}(v_{\max}, u) + 1
=
\Distance_{G}(v_{\max}, v)$,
establishing the step of the secondary induction.
\end{proof}
\par\fussy

By plugging
$d = D$
into \Lem{}~\ref{lemma:module-rand-phase}, we obtain the following corollary.

\begin{corollary} \label{corollary:module-rand-phase-endgame}
Suppose that node $v_{\max}$ resets
$v_{\max}.\varFlag \gets 0$
in round $t$.
Then, all nodes
$v \in V$ \\
(1)
set
$v.\varStep \gets D + 1$
concurrently in round
$t + D$; \\
(2)
set
$v.\varStep \gets D + 2$
concurrently in round
$t + D + 1$;
and \\
(3)
start the next phase concurrently in round
$t + D + 2$.
\end{corollary}

\subsubsection{Implementing Module \ModuleCompete{}.}
Consider the execution of module \ModuleCompete{} in a phase $\pi$ and let
$U \subseteq V$
be the set of nodes that are still undecided at the beginning of $\pi$.
The implementation of \ModuleCompete{} is based on a binary variable,
denoted by
$v.\varCandidate \in \{ 0, 1 \}$,
that each node
$v \in U$
maintains, indicating that $v$ is still a candidate to join $\InMIS$ during
$\pi$.
When $\pi$ begins, $v$ sets
$v.\varCandidate \gets 1$;
then, $v$ proceeds by participating in a sequence of random \emph{trials}
that continues as long as
$v.\varCandidate = 1$
and
$v.\varStep \leq D$
(recall that $v.\varStep$ is the variable that controls the deterministic
suffix of module \ModuleRandPhase{}).
Each trial consists of two rounds:
in the first round, $v$ tosses a fair coin, denoted by
$C_{v} \in_{r} \{ 0, 1 \}$;
in the second round, $v$ computes the indicator
$I_{C} = \bigvee_{u \in \Neighbors^{+}_{G(U)}(v) \, : \, u.\varCandidate = 1} C_{u}$.
If
$C_{v} = 0$
and
$I_{C} = 1$,
then $v$ resets
$v.\varCandidate \gets 0$;
otherwise, $v.\varCandidate$ remains $1$.

If $v.\varCandidate$ is still $1$ when $v.\varStep$ is incremented to
$v.\varStep \gets D + 1$,
then $v$ joins $\InMIS$.
This is sensed in the subsequent round by $v$'s undecided neighbors that join
$\OutMIS$ in response.
Notice that by Corollary~\ref{corollary:module-rand-phase-endgame}, all nodes
increment the $\varStep$ variables concurrently to
$D + 1$
and then to
$D + 2$,
hence nodes may join $\InMIS$ and $\OutMIS$ only during the penultimate and
ultimate rounds, respectively, of phase $\pi$.

We now turn to analyze \ModuleCompete{} during phase $\pi$.
Assume for the sake of the analysis that a node
$v \in U$
keeps on participating in the trials in a ``vacuous'' manner, tossing the
$C_{v}$ coins in vain, even after
$v.\varCandidate \gets 0$,
until
$v.\varStep \gets D + 1$;
this has no influence on the nodes that truly participate in the trials as
the trials' outcome is not influenced by any node
$v \in U$
with
$v.\varCandidate = 0$.

Recall that the guarantees of \ModuleRandPhase{} ensure that at least
$c_{0} \log n$
rounds have elapsed in phase $\pi$ whp before node $v_{\max}$ resets
$v_{\max}.\varFlag \gets 0$,
where $c_{0}$ is an arbitrarily large constant;
condition hereafter on this event.
Moreover, $v_{\max}$ starts to increment variable $v_{\max}.\varStep$ only
after
$v_{\max}.\varFlag \gets 0$.
Therefore, when a node
$v \in U$
sets
$v.\varStep \gets D + 1$,
we know that at least
$c_{0} \log n$
rounds have already elapsed in phase $\pi$ which means that the undecided
nodes participate in at least
$\tau = \lfloor c_{0} / 2 \rfloor \log n$
trials during $\pi$.

For a node
$v \in U$,
let $C_{v}^{i}$ denote the value of the coin $C_{v}$ tossed by $v$ in trial
$i = 1, \dots, \tau$.
Let $v.\varCandidate^{i}$ denote the value of the variable $v.\varCandidate$ at
the beginning of trial
$i = 1, \dots, \tau$
and based on that, define the random variable
$Z(v)
=
\sum_{i = 1}^{\tau} 2^{\tau - i} \cdot C_{v}^{i}$.
Module \ModuleCompete{} is designed so that a node
$v \in V$
with
$v.\varCandidate^{i} = 1$
resets
$v.\varCandidate \gets 0$
in trial
$1 \leq i \leq \tau$
if and only if
(I)
$C_{v}^{i} = 0$;
and
(II)
there exists a node
$u \in \Neighbors_{G(U)}(v)$
such that
$u.\varCandidate^{i} = 1$
and
$C_{u}^{i} = 1$.
We conclude by the definition of $Z(v)$ that $v$ joins $\InMIS$ if and only if
$Z(v) \geq Z(u)$
for all nodes
$u \in \Neighbors_{G(U)}(v)$.
Moreover, a node
$u \notin \InMIS$
joins $\OutMIS$ if and only if there exists a node
$v \in \Neighbors_{G(U)}(u)$
that joins $\InMIS$ in the previous round (this holds deterministically).

To complete the analysis of \ModuleCompete{}, we fix a node
$v \in U$
and prove that
(1)
$Z(v) \neq Z(u)$
for all nodes
$u \in \Neighbors_{G(U)}(u)$
whp;
and
(2)
$\Pr \left( \bigwedge_{u \in W} [Z(v) > Z(u)] \right)
\geq
\Omega \left( \frac{1}{|W| + 1} \right)$
for every node subset
$W \subseteq U - \{ v \}$.
To this end, notice that the random variables $Z(u)$,
$u \in U$,
are independent and distributed uniformly over the (discrete) set
$\{ 0, 1, \dots, 2^{\tau} - 1 \}$.
This means that
$\Pr(Z(u) = Z(u')) = 2^{-\tau} = 1 / n^{\lfloor c_{0} / 2 \rfloor}$
for any two distinct nodes
$u, u' \in U$.
Recalling that $c_{0}$ is an arbitrarily large constant, we conclude, by the
union bound, that the random variables $Z(u)$,
$u \in U$,
are mutually distinct whp, thus establishing (1).
Conditioning on that, (2) follows as the random variables $Z(u)$,
$u \in U$,
are identically distributed.

\subsubsection{Implementing Module \ModuleDetectMIS{}.}
The implementation of module \ModuleDetectMIS{} is rather straightforward:
In every round, each $\InMIS$ node
$v \in V$
picks a temporary (not necessarily unique) identifier uniformly at random from
$[k]$ for a constant
$k \geq 2$.
An $\OutMIS$ node
$u \in V$
with no neighboring $\InMIS$ node is detected as $u$ does not sense any
temporary identifier in its (inclusive) neighborhood (this happens with
probability $1$).
An $\InMIS$ node $v$ with a neighboring $\InMIS$ node is detected when $v$
senses a temporary identifier different from its own, an event that occurs
with probability at least
$1 - 1 / k$.

\subsection{Algorithm \AlgLE{}}
\label{section:algorithm-LE}
The LE algorithm \AlgLE{} share a few design features with \AlgMIS{} that are
presented in this section independently for the sake of completeness.
For clarity of the exposition, \AlgLE{} is presented in a procedural style;
converting it to a randomized state machine with
$O (D)$
states is straightforward.
The algorithm is designed assuming that the execution starts concurrently at
all nodes;
this assumption is plausible due to \Thm{}~\ref{theorem:module-restart} and
given the algorithm's fault detection guarantees (described in the sequel).
Algorithm \AlgLE{} progresses in synchronous \emph{epochs}, where every epoch
lasts for $D$ rounds.
Each node maintains the round number within the current epoch and invokes
\ModuleRestart{} if an inconsistency with one of its neighbors regarding this
round number is detected.

The execution of Algorithm \AlgLE{} starts with a \emph{computation stage},
followed by a \emph{verification stage}.
The computation stage is guaranteed to elect exactly one leader whp;
it runs for
$O (\log n)$
epochs in expectation and whp.
The verification stage starts once the computation stage halts and continues
indefinitely thereafter.
Its role is to verify that the configuration is correct (i.e., the graph
includes exactly one leader).
During the verification stage, a faulty configuration is detected in each
epoch (independently) with a positive constant probability, in which case,
\ModuleRestart{} is invoked and the execution of \AlgLE{} starts from scratch
once \ModuleRestart{} is exited.
Recalling that the execution of \ModuleRestart{} takes
$O (D)$
rounds, one concludes by standard probabilistic arguments that \AlgLE{}
stabilizes within
$O (D \log n)$
rounds in expectation and whp, thus establishing \Thm{}~\ref{theorem:LE}.

\subsubsection{The Computation Stage.}
During the computation stage, algorithm \AlgLE{} runs two modules, denoted by
\ModuleRandCount{} and \ModuleElect{}.
Module \ModuleRandCount{} implements a ``randomized counter'' that signals the
nodes when $X$ epochs have elapsed since the beginning of the computation
stage, where $X$ is a random variable that satisfies
(1)
$X \leq O (\log n)$
in expectation and whp;
and
(2)
$X \geq c_{0} \log n$
whp for a constant
$c_{0} > 0$
that can be made arbitrarily large.
Upon receiving this signal from \ModuleElect{}, the nodes halt the computation
stage (and start the verification stage).

To implement module \ModuleRandCount{}, each node
$v \in V$
maintains a binary variable, denoted by
$v.\varFlag \in \{ 0, 1 \}$,
that is set initially to
$v.\varFlag \gets 1$.
At the beginning of each epoch, if $v.\varFlag$ is still $1$, then $v$
tosses a (biased) coin and resets
$v.\varFlag \gets 0$
with probability
$0 < p_{0} < 1$,
where
$p_{0} = p_{0}(c_{0})$
is a constant determined by $c_{0}$.
The $D$ rounds of the epoch are now employed to allow (all) the nodes to
compute the indicator
$I_{\varFlag} = \bigvee_{u \in V} u.\varFlag$.
If
$I_{\varFlag} = 0$,
then the computation stage is halted.
The correctness of \ModuleRandCount{} follows from
\Obs{}~\ref{observation:max-geometric}.

The role of module \ModuleElect{}, that runs in parallel to
\ModuleRandCount{}, is to elect exactly one leader whp.
The implementation of \ModuleElect{} is based on a binary variable, denoted by
$v.\varCandidate \in \{ 0, 1 \}$,
maintained by each node
$v \in V$,
that indicates that $v$ is still a candidate to be elected as a leader.
Initially, $v$ sets
$v.\varCandidate \gets 1$.
At the beginning of each epoch, if $v.\varCandidate$ is still $1$, then $v$
tosses a fair coin, denoted by
$C_{v} \in_{r} \{ 0, 1 \}$.
The $D$ rounds of the epoch are then employed to allow $v$ (and all other
nodes) to compute the indicator
$I_{C} = \bigvee_{u \in V \, : \, u.\varCandidate = 1} C_{u}$.
If
$C_{v} = 0$
and
$I_{C} = 1$,
then $v$ resets
$v.\varCandidate \gets 0$;
otherwise
$v.\varCandidate$ remains $1$.
If $v.\varCandidate$ is still $1$ when the computation stage comes to a halt
(recall that this event is determined by module \ModuleRandCount{}), then $v$
marks itself as a leader.

To see that module \ModuleElect{} is correct, let $v.\varCandidate^{i}$ and
$C_{v}^{i}$ denote the values of variable $v.\varCandidate$ and of coin
$C_{v}$, respectively, at the beginning of epoch $i$ for each node
$v \in V$.
Notice that if
$v.\varCandidate^{i} = 1$
and
$v.\varCandidate^{i + 1} = 0$,
then there must exist a node
$u \in V$
such that
$u.\varCandidate^{i} = 1$
and
$C_{u}^{i} = 1$,
which implies that
$u.\varCandidate^{i + 1} = 1$.
Therefore, at least one node
$v \in V$
survives as a candidate with
$v.\varCandidate = 1$
at the end of each epoch.

Recall that the computation stage, and hence also module \ModuleElect{}, lasts
for at least
$c_{0} \log n$
epochs whp, where $c_{0}$ is an arbitrarily large constant;
condition hereafter on this event.
Given two nodes
$u, v \in V$,
the probability that
$C_{u}^{i} = C_{v}^{i}$
for
$i = 1, \dots, c_{0} \log n$
is up-bounded by
$2^{-c_{0} \log n} = 1 / n^{c_{0}}$.
Observing that if
$C_{u}^{i} \neq C_{v}^{i}$,
then either
$u.\varCandidate^{i + 1} = 0$
or
$v.\varCandidate^{i + 1} = 0$,
and recalling that $c_{0}$ is an arbitrarily large constant, we conclude, by the
union bound, that no two nodes survive as candidates when \ModuleElect{} halts
whp, thus satisfying the promise of the computation stage.

\subsubsection{The Verification Stage.}
During the verification stage, algorithm \AlgLE{} runs a module denoted by
\ModuleDetectLE{}.
This module is designed to detect configurations that include zero leaders and
configurations that include at least two leaders;
the former task is performed deterministically (and thus succeeds with
probability $1$), whereas the latter relies on a (simple) probabilistic tool
and succeeds with probability at least $p$, where
$0 < p < 1$
is a constant that can be made arbitrarily large.

Module \ModuleDetectLE{} is implemented as follows.
If a node
$v \in V$
is marked as a leader, then at the beginning of each epoch, $v$ picks
a temporary (not necessarily unique) identifier $\mathtt{id}_{v}$ uniformly at
random from $[k]$, where $k$ is a positive constant integer.
The $D$ rounds of the epoch are then employed to verify that there is exactly
one temporary identifier in the graph (in the current epoch).
To this end, each node
$u \in V$
encodes, in its state, the first temporary identifier
$j \in [k]$
that $u$ encounters during the epoch (either by picking j as $u$'s own
temporary identifier or by sensing $j$ in its neighbors' states) and invokes
module \ModuleRestart{} if it encounters any temporary identifier
$j' \in [k] - \{ j \}$;
if $u$ does not encounter any temporary identifier until the end of the epoch,
then it also invokes \ModuleRestart{}.
This ensures that
(1)
if no node is marked as a leader, then all nodes invoke \ModuleRestart{}
(deterministically);
and
(2)
if two (or more) nodes are marked as leaders, then \ModuleRestart{} is invoked
by some nodes with probability at least
$1 - 1 / k$.
The promise of the verification stage follows as $k$ can be made arbitrarily
large.

\subsection{Module \ModuleRestart{}}
\label{section:module-restart}
In this section, we implement module \ModuleRestart{} and establish
\Thm{}~\ref{theorem:module-restart}.
The module consists of
$2 D + 1$
states denoted by
$\sigma(0), \sigma(1), \dots, \sigma(2 D)$,
where states $\sigma(0)$ and $\sigma(2 D)$ play the role of
\ModuleRestart{}-\texttt{entry} and \ModuleRestart{}-\texttt{exit},
respectively.
For a node
$v \in V$,
we subsequently denote the state in which $v$ resides at time $t$ by
$q^{t}(v)$ and the set of states sensed by $v$ at time $t$ by
$S^{t}(v) = \{ q^{t}(u) \mid u \in \Neighbors^{+}(v) \}$;
we also denote the set of all node states by
$Q^{t} = \{ q^{t}(v) \mid v \in V \}$.
The implementation of module \ModuleRestart{} at node $v$ obeys the following
three rules:
\begin{itemize}

\item
If
$S^{t}(v) \cap \{ \sigma(i) \mid 0 \leq i \leq 2 D \} \neq \emptyset$
and
$S^{t}(v) \nsubseteq \{ \sigma(i) \mid 0 \leq i \leq 2 D \}$,
then
$q^{t + 1}(v) \gets \sigma(0)$.

\item
If
$S^{t}(v) \subseteq \{ \sigma(i) \mid 0 \leq i \leq 2 D \}$
and
$S^{t}(v) \neq \{ \sigma(2 D) \}$,
then
$q^{t + 1}(v) \gets \sigma(i_{\min} + 1)$,
where
$i_{\min} = \min \{ i : \sigma(i) \in S^{t}(v) \}$.

\item
If
$S^{t}(v) = \{ \sigma(2 D) \}$,
then
$q^{t + 1}(v) \gets q^{*}_{0}$.

\end{itemize}

\subsubsection{Analysis.}
We now turn to establish \Thm{}~\ref{theorem:module-restart}, starting with the
following two observations.

\begin{observation} \label{observation:module-restart:mix}
If
$Q^{t} \cap \{ \sigma(i) \mid 0 \leq i \leq 2 D \} \neq \emptyset$
and
$Q^{t} \nsubseteq \{ \sigma(i) \mid 0 \leq i \leq 2 D \}$,
then there exists a node
$v \in V$
that enters \ModuleRestart{} in round $t$ so that
$q^{t + 1}(v) = \sigma(0)$.
\end{observation}

\begin{observation} \label{observation:module-restart-progress}
If
$Q^{t} \subseteq \{ \sigma(i) \mid 0 \leq i \leq 2 D \}$
and
$Q^{t} \neq \{ \sigma(2 D) \}$,
then
$\min \{ i : \sigma(i) \in Q^{t + 1} \}
=
\min \{ i : \sigma(i) \in Q^{t} \} + 1$.
\end{observation}

By combining \Obs{} \ref{observation:module-restart:mix} and
\ref{observation:module-restart-progress}, we conclude that if
$Q^{t_{0}} \cap \{ \sigma(i) \mid 0 \leq i \leq 2 D \} \neq \emptyset$,
then there exists a time
$t_{0} \leq t \leq t_{0} + O (D)$
such that either
(1)
all nodes exit \ModuleRestart{}, concurrently, at time $t$;
or
(2)
$\sigma(0) \in Q^{t}$.
Therefore, to establish \Thm{}~\ref{theorem:module-restart}, it suffices to
prove that if
$\sigma(0) \in Q^{t_{0}}$,
then there exists a time
$t_{0} \leq t \leq t_{0} + O (D)$
such that
all nodes exit \ModuleRestart{}, concurrently, at time $t$.
This is done based on the following three lemmas.

\begin{lemma} \label{lemma:module-restart:growing-ball-0}
Consider a node
$v \in V$
and suppose that
$q^{t}(v) = \sigma(0)$.
Then,
$\{ q^{t + d}(u) \mid \Distance_{G}(u, v) \leq d \}
\subseteq
\{ \sigma(j) \mid 0 \leq j \leq d \}$
for every
$0 \leq d \leq D$.
\end{lemma}
\begin{proof}
By induction on
$d = 0, 1, \dots, D$.
The assertion holds trivially for
$d = 0$,
so assume that the assertion holds for
$d - 1$
and consider a node
$u \in V$
whose distance from $v$ is
$\Distance_{G}(u, v) = d$.
Let $u'$ be the node that precedes $u$ along a shortest
$(v, u)$-path in $G$.
Since
$\Distance_{G}(v, u') = d - 1$,
it follows by the inductive hypothesis that
$q^{t + d - 1}(u') \in \{ \sigma(j) \mid 0 \leq j \leq d - 1 \}$.
As
$q^{t + d - 1}(u') \in S^{t + d - 1}(u)$,
we conclude by the design of \ModuleRestart{} that
$q^{t + d}(u) \in \{ \sigma(j) \mid 0 \leq j \leq d \}$,
thus establishing the assertion.
\end{proof}

\begin{lemma} \label{lemma:module-restart:bound-indices}
Assume that
$Q^{t} \subseteq \{ \sigma(j) \mid 0 \leq j \leq D \}$
and let
$j_{\min} = \min \{ j : \sigma(j) \in Q^{t} \}$.
Then,
$Q^{t + h} \subseteq \{ \sigma(i) \mid j_{\min} + h \leq i \leq D + h \}$
for every
$0 \leq h \leq D$.
\end{lemma}
\begin{proof}
By induction on
$h = 0, 1, \dots, D$.
The assertion holds trivially for
$h = 0$,
so assume that the assertion holds for
$h - 1$
and consider a node
$v \in V$.
The inductive hypothesis guarantees that
$S^{t + h - 1}(v)
\subseteq
\{ \sigma(i) \mid j_{\min} + h - 1 \leq i \leq D + h - 1 \}$.
The assertion follows by the design of \ModuleRestart{} ensuring that
$q^{t + h}(v) = \sigma(i_{\min} + 1)$,
where
$i_{\min} = \min \{ i : \sigma(i) \in S^{t + h - 1}(v)\}$.
\end{proof}

\begin{lemma} \label{lemma:module-restart:growing-ball-exact}
Assume that
$Q^{t} \subseteq \{ \sigma(j) \mid 0 \leq j \leq D \}$.
Let
$j_{\min} = \min \{ j : \sigma(j) \in Q^{t} \}$
and let $v_{\min}$ be a node with
$q^{t}(v_{\min}) = \sigma(j_{\min})$.
Then,
$\{ q^{t + d}(v) \mid \Distance_{G}(v_{\min}, v) \leq d \}
=
\{ \sigma(j_{\min} + d) \}$
for every
$0 \leq d \leq D$.
\end{lemma}
\begin{proof}
By induction on
$d = 0, 1, \dots, D$.
The assertion holds trivially for
$d = 0$,
so assume that the assertion holds for
$d - 1$
and consider a node
$v \in V$
whose distance from $v_{\min}$ is
$\Distance_{G}(v_{\min}, v) = d$.
Let $v'$ be the node that precedes $v$ along a shortest
$(v_{\min}, v)$-path in $G$.
Since
$\Distance_{G}(v_{\min}, v') = d - 1$,
it follows by the inductive hypothesis that
$q^{t + d - 1}(v') = \sigma(j_{\min} + d - 1)$.
\Lem{}~\ref{lemma:module-restart:bound-indices} ensures that
$S^{t + d - 1}(v)
\subseteq
\{ \sigma(i) \mid j_{\min} + d - 1 \leq i \leq D + d - 1 < 2 D \}$,
hence
$\min \{ i : \sigma(i) \in S^{t + d - 1}(v) \} = j_{\min} + d - 1$
and
$q^{t + d}(v) = \sigma(j_{\min} + d)$,
thus establishing the assertion.
\end{proof}

We are now ready to complete the proof of \Thm{}~\ref{theorem:module-restart}.
Consider a node
$v \in V$
that satisfies
$q^{t_{0}}(v) = \sigma(0)$.
By employing \Lem{}~\ref{lemma:module-restart:growing-ball-0} with
$d = D$,
we deduce that
$Q^{t_{0} + D} \subseteq \{ \sigma(j) \mid 0 \leq j \leq D \}$.
Therefore, we can employ \Lem{}~\ref{lemma:module-restart:growing-ball-exact}
with
$d = D$
to conclude that there exists an index
$D \leq i \leq 2 D$
such that
$Q^{t_{0} + 2 D} = \{ \sigma(i) \}$.
From time
$t_{0} + 2 D$
onwards,
all nodes ``progress in synchrony'' until time
$t_{0} + 2 D + 2 D - i = t_{0} + 4 D - i$
at which we get
$Q^{t_{0} + 4 D - i} = \{ \sigma(2 D) \}$.
Thus, all nodes exit \ModuleRestart{}, concurrently, in round
$t_{0} + 4 D - i \leq t_{0} + 3 D$.

\section{Synchronizer}
\label{section:synchronizer}
Consider a distributed task $\mathcal{T}$, restricted to $D$-bounded diameter
graphs, and let
$\Pi
=
\left\langle
Q, Q_{\mathcal{O}}, \omega, \delta
\right\rangle$
be a synchronous self-stabilizing algorithm for $\mathcal{T}$ whose
stabilization time on $n$-node instances is bounded by
$f(n, D)$
in expectation and whp.
Our goal in this section is to lift the synchronous schedule assumption, thus
establishing Corollary~\ref{corollary:synchronizer}.
Specifically, we employ the self-stabilizing AU algorithm \AlgAU{} promised in
\Thm{}~\ref{theorem:AU}, combined with the ideas behind the
non-self-stabilizing SA transformer of \cite{EmekW2013stone} (see also
\cite{AfekEK2018synergy}), to develop a synchronizer that converts $\Pi$ into
a self-stabilizing algorithm
$\Pi^{*} =
\left\langle
Q^{*}, Q^{*}_{\mathcal{O}}, \omega^{*}, \delta^{*}
\right\rangle$
for $\mathcal{T}$ with state space
$|Q^{*}| \leq O (D \cdot |Q|^{2})$
whose stabilization time on $n$-node instances is bounded by
$f(n, D) + O (D^{3})$
in expectation and whp for any (arbitrarily asynchronous) schedule.

Let $K$ be the cyclic group corresponding to the AU clock values.
Let $T$ and $T_{K}$ be the state set and output state set, respectively, of
\AlgAU{}.
The state set $Q^{*}$ of $\Pi^{*}$ is defined to be the Cartesian product
$Q^{*} = Q \times Q \times T$.
We also define
$Q^{*}_{\mathcal{O}} = \{ Q_{\mathcal{O}} \times Q \times T_{K} \}$
and for each output $\Pi^{*}$-state
$s = (q, q', \nu) \in Q^{*}_{\mathcal{O}}$,
define
$\omega^{*}(s) = \omega(q)$.

Consider a node
$v \in V$
residing in a state
$s = (q, q', \nu) \in Q^{*}$
of $\Pi^{*}$.
The state transition function $\delta^{*}$ of $\Pi^{*}$ is designed so that
$\Pi^{*}$ simulates the operation of \AlgAU{}, encoding \AlgAU{}'s current
state in the third coordinate of $s$.
Once \AlgAU{} has stabilized, $\Pi^{*}$ uses the first two coordinates of $s$
to simulate the operation of $\Pi$ every time \AlgAU{} advances its clock
value, interpreting $q$ and $q'$ as $v$'s current and previous $\Pi$-states,
respectively.

More formally, suppose that node $v$ is activated at time $t$ and that
\AlgAU{} advances its clock value by changing its output state from
$\nu \in T_{K}$
to
$\nu' \in T_{K}$
in step $t$.
When this happens, node $v$ moves from $\Pi^{*}$-state
$s = (q, q', \nu) \in Q^{*}$
to $\Pi^{*}$-state
$s' = (p, q, \nu') \in Q^{*}$,
where the $\Pi$-state $p$ is determined according to the following mechanism:
Let
$\mathcal{S}_{v, \Pi}^{t} \in \{ 0, 1 \}^{Q}$
be the simulated $\Pi$-signal of $v$ at time $t$ defined by
setting
$\mathcal{S}_{v, \Pi}^{t}(r) = 1$,
$r \in Q$,
if and only if $v$ senses at time $t$ at least one $\Pi^{*}$-state of the form
$(r, \cdot, \nu)$
or
$(\cdot, r, \nu')$.
The $\Pi$-state $p$ is then determined by applying the state transition
function of $\Pi$ to $q$ and
$\mathcal{S}_{v, \Pi}^{t}$,
that is, $p$ is picked uniformly at random from
$\delta \left( q, \mathcal{S}_{v, \Pi}^{t} \right)$.

\section{Related Work and Discussion}
\label{section:related-work-discussion}
The algorithmic model considered in the current paper is a restricted version
of the SA model introduced by Emek and Wattenhofer \cite{EmekW2013stone} and
studied subsequently by Afek et al.\
\cite{AfekEK2018selecting, AfekEK2018synergy}
and Emek and Uitto \cite{EmekU2020dynamic}.
Specifically, the communication scheme in the latter model relies on
asynchronous message passing, thus enhancing the power of the adversarial
scheduler by allowing it to determine not only the node activation pattern,
but also the time delay of each transmitted message.
Whether our algorithmic results can be modified to work with such a (stronger)
scheduler is left as an open question.
The reader is referred to
\cite{AfekEK2018selecting, AfekEK2018synergy}
for a discussion of various other aspects of the SA model and its variants.

The communication scheme of the SA model can be viewed as an asynchronous
version of the \emph{set-broadcast (SB)} communication model of
\cite{HellaJKLLLSV2015weak}.
It is also closely related to the \emph{beeping} model
\cite{CornejoK2010deploying, FluryW2010slotted},
where in every (synchronous) round, each node either listens or beeps and a
listening node receives a binary signal indicating whether at least one of its
neighbors beeps in that round.
In particular, the communication scheme used in the current paper can be
regarded as an extension of the beeping model (with no sender collision
detection) to asynchronous executions over multiple (yet, a fixed number of)
channels.

Most of the algorithms developed in the beeping model literature consider a
fault free environment.
Two exceptions are the self-stabilizing MIS algorithms developed by Afek et
al.~\cite{AfekABCHK2011beeping} and Scott et al.~\cite{ScottJX2013feedback} that
work under the assumption that the nodes know an approximation of $n$ and that
this parameter cannot be modified by the adversary.\footnote{%
In \cite{ScottJX2013feedback}, the knowledge of $n$ is implicit and is
only required for bounding the initial values in the node's registers.}
In contrast, our algorithmic model is inherently size-uniform as the nodes
cannot even encode (any function of) $n$ in their internal memory.

A beeping algorithm that is more closely related to the computational
limitations of our model is that of Gilbert and Newport
\cite{GilbertN2015computational} for LE in complete graphs.
This algorithm is implemented by nodes with constant size internal memory,
hence it can be viewed as a SA algorithm with a single message type.
In fact, one of the techniques used in the current paper for implementing a
probabilistic counter resembles a technique used also in
\cite{GilbertN2015computational}.
Notice though that the algorithm of \cite{GilbertN2015computational} is not
only restricted to complete graphs, but also requires a synchronous schedule
and cannot cope with transient faults;
in this regard, it is less robust than our LE algorithm.

The AU task was introduced by Couvreur et
al.~\cite{CouvreurFG1992asynchronous} as a fundamental primitive for
asynchronous systems.
Shortly after, Awerbuch et al.~\cite{AwerbuchKMPV1993time} observed that this
task captures the essence of constructing a self-stabilizing synchronizer and
developed an anonymous size-uniform self-stabilizing AU algorithm that
stabilizes in
$O (D)$
time, albeit with an unbounded state space.
By incorporating a reset module into their algorithm, Awerbuch et al.\
obtained a self-stabilizing AU algorithm with a bounded state space and the
same asymptotic stabilization time, however, the reset module requires unique
node IDs and/or the knowledge of $n$ (or an approximation thereof), which
means in particular that its state space is
$\Omega (\log n)$;
it also relies on unicast communication.

Since then, the AU task has been extensively investigated in different
computational models and for a variety of graph classes
\cite{BoulinierPV2004graph, BoulinierPV2005synchronous, BoulinierPV2006time,
DevismesP2012efficiency, DevismesJ2019self}.
For general graphs, Boulinier et al.~\cite{BoulinierPV2004graph} developed a
self-stabilizing AU algorithm that can be implemented under a set-broadcast
communication model (very similar to the communication model used in the
current paper).
When applied to a graph $G$, the state space and stabilization time bounds of
their algorithm are linear in
$C_{G} + T_{G}$,
where $C_{G}$ is the minimum longest cycle length among all cycle bases of $G$
(or $2$ if $G$ is cycle free) and $T_{G}$ is the length of the longest
chordless cycle of $G$ (or $2$ if $G$ is cycle free).
While $C_{G}$ is up-bounded by
$O (D)$
for every graph $G$ (in particular, all cycles of the fundamental cycle basis
of a breadth-first search tree are of length
$O (D)$), the performance of the AU algorithm of Boulinier et al.\ cannot be
directly compared to the performance of our AU algorithm due to the dependency
of the former on $T_{G}$:
on the one hand, there are graphs of linear diameter in which
$T_{G} = O (1)$;
on the other hand, there are graphs of constant diameter in which
$T_{G} = \Omega (n)$.

\section*{Acknowledgments}
We are grateful to Shay Kutten and Yoram Moses for helpful discussions.

\clearpage
\bibliographystyle{alpha}
\bibliography{references}

\newcommand{\etalchar}[1]{$^{#1}$}
\begin{thebibliography}{AKM{\etalchar{+}}93}

\bibitem[AAB{\etalchar{+}}11]{AfekABCHK2011beeping}
Yehuda Afek, Noga Alon, Ziv Bar{-}Joseph, Alejandro Cornejo, Bernhard Haeupler,
  and Fabian Kuhn.
\newblock Beeping a maximal independent set.
\newblock In David Peleg, editor, {\em Distributed Computing - 25th
  International Symposium, {DISC} 2011, Rome, Italy, September 20-22, 2011.
  Proceedings}, volume 6950 of {\em Lecture Notes in Computer Science}, pages
  32--50. Springer, 2011.

\bibitem[ABI86]{AlonBI1986fast}
Noga Alon, L{\'{a}}szl{\'{o}} Babai, and Alon Itai.
\newblock A fast and simple randomized parallel algorithm for the maximal
  independent set problem.
\newblock {\em J. Algorithms}, 7(4):567--583, 1986.

\bibitem[ADDP19]{self_2019}
Karine Altisen, St{\'{e}}phane Devismes, Swan Dubois, and Franck Petit.
\newblock {\em Introduction to Distributed Self-Stabilizing Algorithms}.
\newblock Synthesis Lectures on Distributed Computing Theory. Morgan {\&}
  Claypool Publishers, 2019.

\bibitem[AEK18a]{AfekEK2018selecting}
Yehuda Afek, Yuval Emek, and Noa Kolikant.
\newblock Selecting a leader in a network of finite state machines.
\newblock In Ulrich Schmid and Josef Widder, editors, {\em 32nd International
  Symposium on Distributed Computing, {DISC} 2018, New Orleans, LA, USA,
  October 15-19, 2018}, volume 121 of {\em LIPIcs}, pages 4:1--4:17. Schloss
  Dagstuhl - Leibniz-Zentrum f{\"{u}}r Informatik, 2018.

\bibitem[AEK18b]{AfekEK2018synergy}
Yehuda Afek, Yuval Emek, and Noa Kolikant.
\newblock The synergy of finite state machines.
\newblock In Jiannong Cao, Faith Ellen, Luis Rodrigues, and Bernardo Ferreira,
  editors, {\em 22nd International Conference on Principles of Distributed
  Systems, {OPODIS} 2018, December 17-19, 2018, Hong Kong, China}, volume 125
  of {\em LIPIcs}, pages 22:1--22:16. Schloss Dagstuhl - Leibniz-Zentrum
  f{\"{u}}r Informatik, 2018.

\bibitem[AKM{\etalchar{+}}93]{AwerbuchKMPV1993time}
Baruch Awerbuch, Shay Kutten, Yishay Mansour, Boaz Patt{-}Shamir, and George
  Varghese.
\newblock Time optimal self-stabilizing synchronization.
\newblock In S.~Rao Kosaraju, David~S. Johnson, and Alok Aggarwal, editors,
  {\em Proceedings of the Twenty-Fifth Annual {ACM} Symposium on Theory of
  Computing, May 16-18, 1993, San Diego, CA, {USA}}, pages 652--661. {ACM},
  1993.

\bibitem[Awe85]{Awerbuch1985complexity}
Baruch Awerbuch.
\newblock Complexity of network synchronization.
\newblock {\em J. ACM}, 32(4):804–823, 1985.

\bibitem[BPV04]{BoulinierPV2004graph}
Christian Boulinier, Franck Petit, and Vincent Villain.
\newblock When graph theory helps self-stabilization.
\newblock In Soma Chaudhuri and Shay Kutten, editors, {\em Proceedings of the
  Twenty-Third Annual {ACM} Symposium on Principles of Distributed Computing,
  {PODC} 2004, St. John's, Newfoundland, Canada, July 25-28, 2004}, pages
  150--159. {ACM}, 2004.

\bibitem[BPV05]{BoulinierPV2005synchronous}
Christian Boulinier, Franck Petit, and Vincent Villain.
\newblock Synchronous vs. asynchronous unison.
\newblock In Ted Herman and S{\'{e}}bastien Tixeuil, editors, {\em
  Self-Stabilizing Systems, 7th International Symposium, {SSS} 2005, Barcelona,
  Spain, October 26-27, 2005, Proceedings}, volume 3764 of {\em Lecture Notes
  in Computer Science}, pages 18--32. Springer, 2005.

\bibitem[BPV06]{BoulinierPV2006time}
Christian Boulinier, Franck Petit, and Vincent Villain.
\newblock Toward a time-optimal odd phase clock unison in trees.
\newblock In Ajoy~Kumar Datta and Maria Gradinariu, editors, {\em
  Stabilization, Safety, and Security of Distributed Systems, 8th International
  Symposium, {SSS} 2006, Dallas, TX, USA, November 17-19, 2006, Proceedings},
  volume 4280 of {\em Lecture Notes in Computer Science}, pages 137--151.
  Springer, 2006.

\bibitem[CFG92]{CouvreurFG1992asynchronous}
Jean{-}Michel Couvreur, Nissim Francez, and Mohamed~G. Gouda.
\newblock Asynchronous unison (extended abstract).
\newblock In {\em Proceedings of the 12th International Conference on
  Distributed Computing Systems, Yokohama, Japan, June 9-12, 1992}, pages
  486--493. {IEEE} Computer Society, 1992.

\bibitem[CK10]{CornejoK2010deploying}
Alejandro Cornejo and Fabian Kuhn.
\newblock Deploying wireless networks with beeps.
\newblock In Nancy~A. Lynch and Alexander~A. Shvartsman, editors, {\em
  Distributed Computing, 24th International Symposium, {DISC} 2010, Cambridge,
  MA, USA, September 13-15, 2010. Proceedings}, volume 6343 of {\em Lecture
  Notes in Computer Science}, pages 148--162. Springer, 2010.

\bibitem[Dij74]{Dijkstra1982}
Edsger~W. Dijkstra.
\newblock Self-stabilizing systems in spite of distributed control.
\newblock {\em Commun. {ACM}}, 17(11):643--644, 1974.

\bibitem[DJ19]{DevismesJ2019self}
St{\'{e}}phane Devismes and Colette Johnen.
\newblock Self-stabilizing distributed cooperative reset.
\newblock In {\em 39th {IEEE} International Conference on Distributed Computing
  Systems, {ICDCS} 2019, Dallas, TX, USA, July 7-10, 2019}, pages 379--389.
  {IEEE}, 2019.

\bibitem[Dol00]{dolev_self-stabilization_2000}
Shlomi Dolev.
\newblock {\em Self-Stabilization}.
\newblock {MIT} Press, 2000.

\bibitem[DP12]{DevismesP2012efficiency}
St{\'{e}}phane Devismes and Franck Petit.
\newblock On efficiency of unison.
\newblock In L{\'{e}}lia Blin and Yann Busnel, editors, {\em 4th Workshop on
  Theoretical Aspects of Dynamic Distributed Systems, {TADDS} '12, Roma, Italy,
  December 17, 2012}, pages 20--25. {ACM}, 2012.

\bibitem[DT11]{DuboisT2011taxonomy}
Swan Dubois and S{\'{e}}bastien Tixeuil.
\newblock A taxonomy of daemons in self-stabilization.
\newblock {\em CoRR}, abs/1110.0334, 2011.

\bibitem[EU20]{EmekU2020dynamic}
Yuval Emek and Jara Uitto.
\newblock Dynamic networks of finite state machines.
\newblock {\em Theor. Comput. Sci.}, 810:58--71, 2020.

\bibitem[EW13]{EmekW2013stone}
Yuval Emek and Roger Wattenhofer.
\newblock Stone age distributed computing.
\newblock In Panagiota Fatourou and Gadi Taubenfeld, editors, {\em {ACM}
  Symposium on Principles of Distributed Computing, {PODC} '13, Montreal, QC,
  Canada, July 22-24, 2013}, pages 137--146. {ACM}, 2013.

\bibitem[FW10]{FluryW2010slotted}
Roland Flury and Roger Wattenhofer.
\newblock Slotted programming for sensor networks.
\newblock In Tarek~F. Abdelzaher, Thiemo Voigt, and Adam Wolisz, editors, {\em
  Proceedings of the 9th International Conference on Information Processing in
  Sensor Networks, {IPSN} 2010, April 12-16, 2010, Stockholm, Sweden}, pages
  24--34. {ACM}, 2010.

\bibitem[GN15]{GilbertN2015computational}
Seth Gilbert and Calvin~C. Newport.
\newblock The computational power of beeps.
\newblock In Yoram Moses, editor, {\em Distributed Computing - 29th
  International Symposium, {DISC} 2015, Tokyo, Japan, October 7-9, 2015,
  Proceedings}, volume 9363 of {\em Lecture Notes in Computer Science}, pages
  31--46. Springer, 2015.

\bibitem[HJK{\etalchar{+}}15]{HellaJKLLLSV2015weak}
Lauri Hella, Matti J{\"{a}}rvisalo, Antti Kuusisto, Juhana Laurinharju, Tuomo
  Lempi{\"{a}}inen, Kerkko Luosto, Jukka Suomela, and Jonni Virtema.
\newblock Weak models of distributed computing, with connections to modal
  logic.
\newblock {\em Distributed Comput.}, 28(1):31--53, 2015.

\bibitem[MRSZ11]{MetivierRSZ2011optimal}
Yves M{\'{e}}tivier, John~Michael Robson, Nasser Saheb{-}Djahromi, and Akka
  Zemmari.
\newblock An optimal bit complexity randomized distributed {MIS} algorithm.
\newblock {\em Distributed Comput.}, 23(5-6):331--340, 2011.

\bibitem[SJX13]{ScottJX2013feedback}
Alex Scott, Peter Jeavons, and Lei Xu.
\newblock Feedback from nature: an optimal distributed algorithm for maximal
  independent set selection.
\newblock In Panagiota Fatourou and Gadi Taubenfeld, editors, {\em {ACM}
  Symposium on Principles of Distributed Computing, {PODC} '13, Montreal, QC,
  Canada, July 22-24, 2013}, pages 147--156. {ACM}, 2013.

\end{thebibliography}

\clearpage
\appendix

\begin{figure}[!t]
{\centering
\LARGE{APPENDIX}
\par}
\end{figure}

\section{A Failed Attempt}
\label{appendix:failed-attempt}
In this section, we present a failed attempt to design a self-stabilizing AU
algorithm based on the design feature of restarting the algorithm when a fault
is detected.
Specifically, the algorithm consists of two components:
the main component is responsible for the liveness condition, controlling the
execution when no faults occur;
the second component is a reset mechanism, responsible for restarting the
execution from a fault free initial configuration when a fault is detected.

Given a constant
$c > 1$,
let
$T = \{ \ell | 0 \leq \ell \leq c D \}$
be the set of turns of the main component and let
$R = \{ R_{i} | 0 \leq i \leq c D \}$
be the set of reset turns.
For a node
$v \in V$,
let $\theta_{v}^{t}$ be the turn of $v$ at time $t$ and let
$\Theta_{v}^{t} = \{ \theta_{u}^{t} \mid u \in \Neighbors^{+}(v) \}$
be the set of turns that $v$ senses at time $t$.
The protocol has three types of state transitions presented from the
perspective of a node
$v \in V$.

\paragraph{State transition of type (ST1).}
The first type of state transitions is equivalent to the type AA transitions
of \AlgAU{}.
Suppose that $v$ is activated at time $t$ and that
$\theta_{v}^{t} = \ell \in T$
and let
$\ell' = \ell + 1 \bmod c D + 1 $.
Then, $v$ performs a type (ST1) transition if
$\Lambda^{t}_{v} \subseteq \{ \ell, \ell' \}$.
This type of state transition updates the turn of $v$ from
$\theta_{v}^{t} = \ell$
to
$\theta_{v}^{t + 1} = \ell'$.

\paragraph{State transition of type (ST2).}
The second type of state transition is applied when $v$ senses a fault.
Specifically, suppose that $v$ is activated at time $t$ and that
$\theta_{v}^{t} = \ell \in T$
and let
$\ell' = \ell + 1 \bmod c D + 1$
and
$\ell'' = \ell - 1 \bmod c D + 1$.
Then,
(1)
if
$\ell \neq 0$,
then $v$ performs a type (ST2) transition if
$\Theta_{v}^{t} \nsubseteq \{ \ell, \ell', \ell'' \}$;
and
(2)
if
$\ell = 0$,
then $v$ performs a type (ST2) transition if
$\Theta_{v}^{t} \nsubseteq \{ \ell, \ell', \ell'', R_{c D} \}$.
This type of state transition updates the turn of $v$ from
$\theta_{v}^{t} = \ell$
to
$\theta_{v}^{t + 1} = R_{0}$.

\paragraph{State transition of type (ST3).}
The third type of state transitions is responsible for the the progress of the
reset mechanism.
Suppose that $v$ is activated at time $t$ and that
$\theta_{v}^{t} = R_{i}$.
Then, $v$ performs a type (ST3) transition if either
(1)
$i \neq c D$
and
$\Theta_{v}^{t} \subseteq \{ R_{j} | i \leq j \leq c D \}$;
or
(2)
$i = c D$
and
$\Theta_{v}^{t} \subseteq \{ R_{c D}, 0 \}$.
This type of state transition updates the turn of $v$ from
$\theta_{v}^{t} = R_{i}$
to
(1)
$\theta_{v}^{t + 1} = R_{i + 1}$
if
$i \neq c D$;
(2)
$\theta_{v}^{t + 1} = 0$
if
$i = c D$.

\subsection*{Counter Example}
Consider the configuration depicted in Figure~\ref{1exa:1},
where
$D = 2$
and assume that
$c = 2$
(the example can be easily adapted to other choices of the constant $c$).
Suppose that node $v_{t - 1}$ is activated in step $t$
for
$t = 1, \dots, 8$.
Notice that \\
(1)
nodes $v_{0}$ and $v_{1}$ do not change their turns; \\
(2)
node $v_{2}$ performs a type (ST2) transition;
and \\
(3)
node $v_{i}$ performs a type (ST3) transition for
$3 \leq i \leq 7$. \\
This means that at time $9$, we reach the configuration depicted in
Figure~\ref{1exa:2}.
As this configuration is equivalent to the configuration at time $0$
up to a node renaming (a rotation of Figure~\ref{1exa:2} in the
counter-clockwise direction), we conclude that the algorithm is in a
live-lock.

\clearpage

\begin{figure}[!t]
{\centering
\LARGE{FIGURES AND TABLES}
\par}
\end{figure}

\begin{table}
\caption{\label{table:AU:state-transition}%
The transition types of \AlgAU{} in step $t$.}
\centering
\begin{tabular}{c|c|c|c}
Type & Pre-transition turn & Post-transition turn & Condition \\
\hline
AA &
$\Able{\ell}$,
$1 \leq |\ell| \leq k$ &
$\Able{\forward^{+1}(\ell)}$ &
$v$ is good and
$\Lambda_{v}^{t} \subseteq \left\{ \ell, \forward^{+1}(\ell) \right\}$ \\
\hline
AF &
$\Able{\ell}$,
$2 \leq |\ell| \leq k$ &
$\Faulty{\ell}$ &
$v \notin \ProtectedV^{t}$
or
$v$ senses turn $\Faulty{\outwards^{-1}(\ell)}$ \\
\hline
FA &
$\Faulty{\ell}$,
$2 \leq |\ell| \leq k$ &
$\Able{\Outwards^{-1}(\ell)}$ &
$\Lambda_{v}^{t} \cap \Outwards^{>}(\ell) = \emptyset$ \\
\end{tabular}
\end{table}

\begin{figure}
{\centering
\includegraphics[scale=0.9]{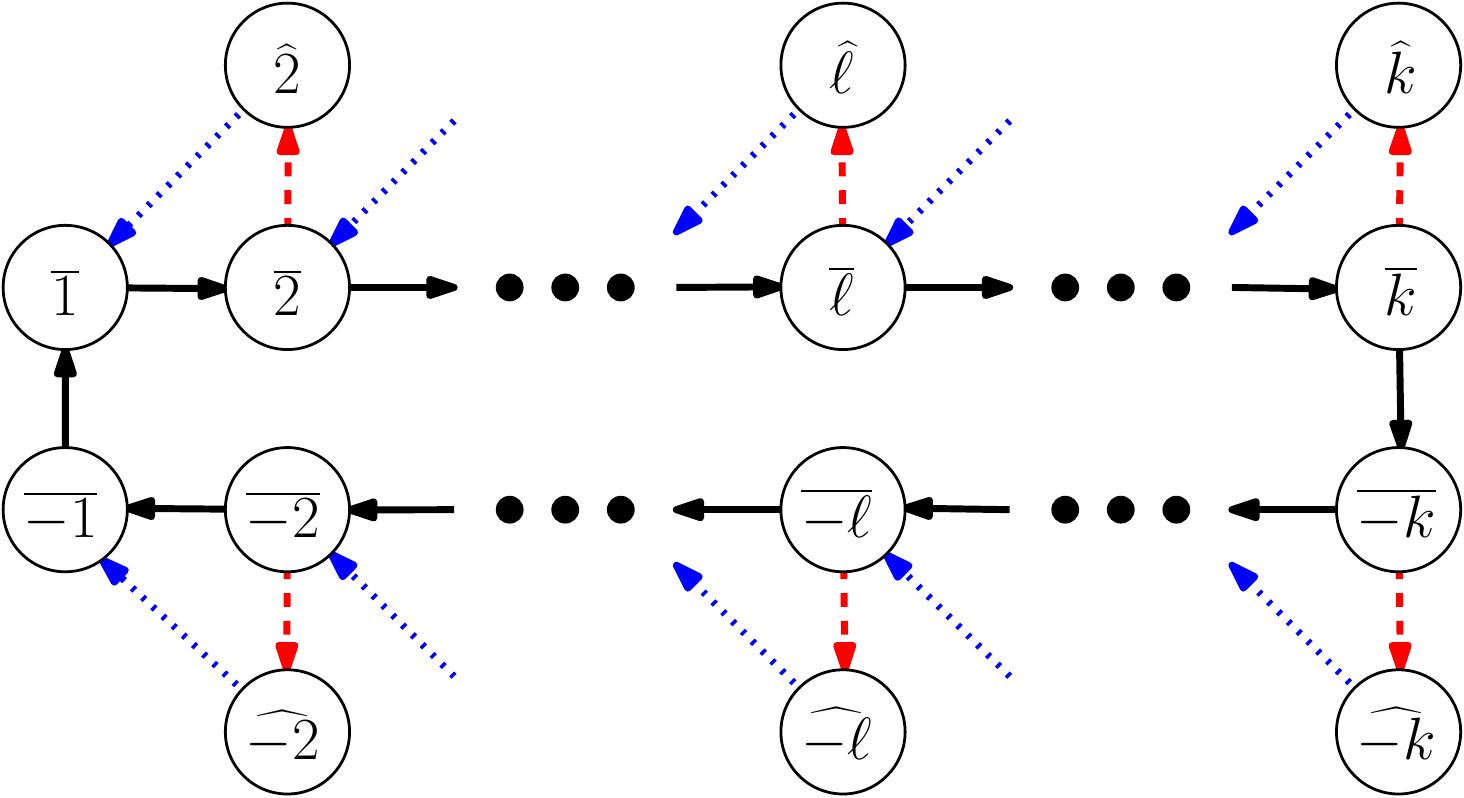}
\par}
\caption{\label{figure:AU:turn-chart}%
The turns of \AlgAU{} and their transition diagram.
The type AA transitions, type AF transitions, and type FA transitions are
depicted by the solid (black) arrows, dashed (red) arrows, and dotted (blue)
arrows, respectively.}
\end{figure}

\begin{figure}
\centering

\subfigure[]{
\begin{tikzpicture}[scale=0.7]	\label{1exa:1}

\def \n {7}
\def \radius {4cm}
\def \margin {10} 

  \node[draw, circle, , minimum size = 1cm] (First-0) at ({360/\n * 0}:\radius) {$0$} ;
  \draw[-, >=latex, thick] ({360/\n * 0+\margin}:\radius) 
    arc ({360/\n * 0+\margin}:{360/\n * (1)-\margin}:\radius);
    
   \node[draw, circle, minimum size = 1cm] (Second-0) at ({360/\n * (1)}:\radius) {$0$};
  \draw[-, >=latex, thick] ({360/\n * (1)+\margin}:\radius) 
    arc ({360/\n * (1)+\margin}:{360/\n * (2)-\margin}:\radius);
    
     \node[draw, circle, minimum size = 1cm] (R_{0}) at ({360/\n * (2)}:\radius) {$R_{0}$};
  \draw[-, >=latex, thick] ({360/\n * (2)+\margin}:\radius) 
    arc ({360/\n * (2)+\margin}:{360/\n * (3)-\margin}:\radius);
    
  \node[draw, circle, minimum size = 1cm] (R_{1}) at ({360/\n * (3)}:\radius) {$R_{1}$};
  \draw[-, >=latex, thick] ({360/\n * (3)+\margin}:\radius) 
    arc ({360/\n * (3)+\margin}:{360/\n * (4)-\margin}:\radius);

  \node[draw, circle, minimum size = 1cm] (R_{2}) at ({360/\n * (4)}:\radius) {$R_{2}$};
  \draw[-, >=latex, thick] ({360/\n * (4)+\margin}:\radius) 
    arc ({360/\n * (4)+\margin}:{360/\n * (5)-\margin}:\radius);

      \node[draw, circle, minimum size = 1cm] (R_{3}) at ({360/\n * (5)}:\radius) {$R_{3}$};
  \draw[-, >=latex, thick] ({360/\n * (5)+\margin}:\radius) 
    arc ({360/\n * (5)+\margin}:{360/\n * (6)-\margin}:\radius);

         \node[draw, circle, minimum size = 1cm] (R_{4}) at ({360/\n * (6)}:\radius) {$R_{4}$};
  \draw[-, >=latex, thick] ({360/\n * (6)+\margin}:\radius) 
    arc ({360/\n * (6)+\margin}:{360/\n * (7)-\margin}:\radius);

	\node[draw, circle, minimum size = 1cm] (mid) at (0,0) {$R_{4}$};
	\node[above, at=(mid.north)] {$v_{0}$};

	\draw[-, >=latex, thick] (mid) -- (First-0);
	\draw[-, >=latex, thick] (mid) -- (Second-0);
	\draw[-, >=latex, thick] (mid) -- (R_{0});
	\draw[-, >=latex, thick] (mid) -- (R_{1});
	\draw[-, >=latex, thick] (mid) -- (R_{2});
	\draw[-, >=latex, thick] (mid) -- (R_{3});
	\draw[-, >=latex, thick] (mid) -- (R_{4});

\node at ({360/\n * 0}:5cm) {$v_{1}$};

  \node at ({360/\n * (1)}:5cm) {$v_{2}$};

  \node  at ({360/\n * (2)}:5cm) {$v_{3}$};

  \node at ({360/\n * (3)}:5cm) {$v_{4}$};

  \node at ({360/\n * (4)}:5cm) {$v_{5}$};

  \node  at ({360/\n * (5)}:5cm) {$v_{6}$};
    
  \node  at ({360/\n * (6)}:5cm) {$v_{7}$};

\end{tikzpicture}
}
\subfigure[]{
\begin{tikzpicture}[scale=0.7]	\label{1exa:2}

\def \n {7}
\def \radius {4cm}
\def \margin {10} 

  \node[draw, circle, , minimum size = 1cm] (First-0) at ({360/\n * 0}:\radius) {$0$} ;
  \draw[-, >=latex, thick] ({360/\n * 0+\margin}:\radius) 
    arc ({360/\n * 0+\margin}:{360/\n * (1)-\margin}:\radius);
    
   \node[draw, circle, minimum size = 1cm] (Second-0) at ({360/\n * (1)}:\radius) {$R_{0}$};
  \draw[-, >=latex, thick] ({360/\n * (1)+\margin}:\radius) 
    arc ({360/\n * (1)+\margin}:{360/\n * (2)-\margin}:\radius);
    
     \node[draw, circle, minimum size = 1cm] (R_{0}) at ({360/\n * (2)}:\radius) {$R_{1}$};
  \draw[-, >=latex, thick] ({360/\n * (2)+\margin}:\radius) 
    arc ({360/\n * (2)+\margin}:{360/\n * (3)-\margin}:\radius);
    
  \node[draw, circle, minimum size = 1cm] (R_{1}) at ({360/\n * (3)}:\radius) {$R_{2}$};
  \draw[-, >=latex, thick] ({360/\n * (3)+\margin}:\radius) 
    arc ({360/\n * (3)+\margin}:{360/\n * (4)-\margin}:\radius);

  \node[draw, circle, minimum size = 1cm] (R_{2}) at ({360/\n * (4)}:\radius) {$R_{3}$};
  \draw[-, >=latex, thick] ({360/\n * (4)+\margin}:\radius) 
    arc ({360/\n * (4)+\margin}:{360/\n * (5)-\margin}:\radius);

      \node[draw, circle, minimum size = 1cm] (R_{3}) at ({360/\n * (5)}:\radius) {$R_{4}$};
  \draw[-, >=latex, thick] ({360/\n * (5)+\margin}:\radius) 
    arc ({360/\n * (5)+\margin}:{360/\n * (6)-\margin}:\radius);

         \node[draw, circle, minimum size = 1cm] (R_{4}) at ({360/\n * (6)}:\radius) {$0$};
  \draw[-, >=latex, thick] ({360/\n * (6)+\margin}:\radius) 
    arc ({360/\n * (6)+\margin}:{360/\n * (7)-\margin}:\radius);

	\node[draw, circle, minimum size = 1cm] (mid) at (0,0) {$R_{4}$};
	\node[above, at=(mid.north)] {$v_{0}$};

	\draw[-, >=latex, thick] (mid) -- (First-0);
	\draw[-, >=latex, thick] (mid) -- (Second-0);
	\draw[-, >=latex, thick] (mid) -- (R_{0});
	\draw[-, >=latex, thick] (mid) -- (R_{1});
	\draw[-, >=latex, thick] (mid) -- (R_{2});
	\draw[-, >=latex, thick] (mid) -- (R_{3});
	\draw[-, >=latex, thick] (mid) -- (R_{4});

\node at ({360/\n * 0}:5cm) {$v_{1}$};

  \node at ({360/\n * (1)}:5cm) {$v_{2}$};

  \node  at ({360/\n * (2)}:5cm) {$v_{3}$};

  \node at ({360/\n * (3)}:5cm) {$v_{4}$};

  \node at ({360/\n * (4)}:5cm) {$v_{5}$};

  \node  at ({360/\n * (5)}:5cm) {$v_{6}$};
    
  \node  at ({360/\n * (6)}:5cm) {$v_{7}$};

\end{tikzpicture}
}
\caption{A live-lock.}
\end{figure}
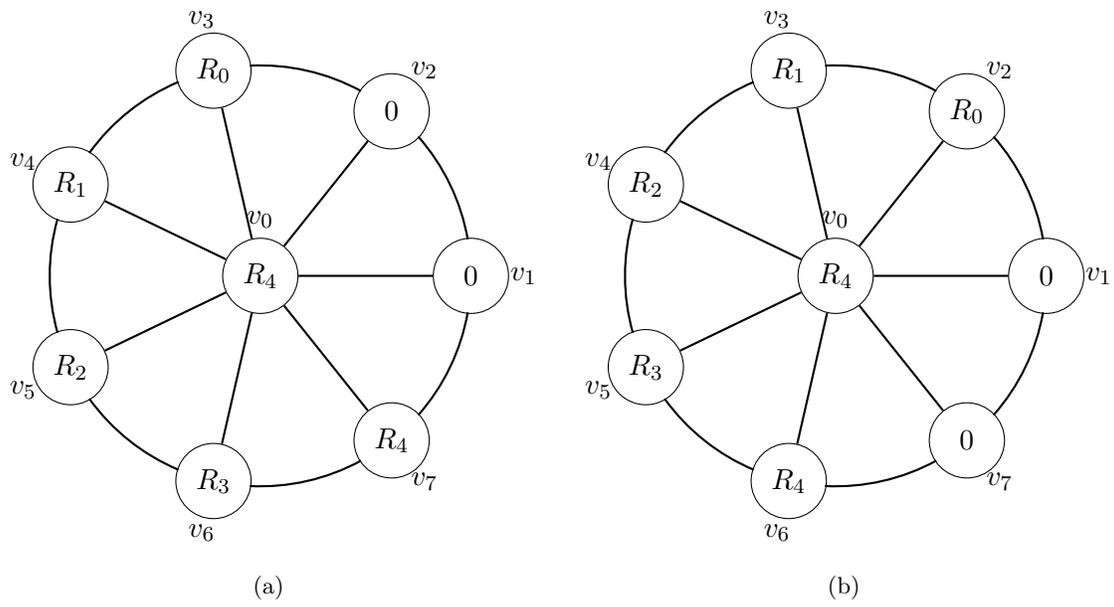

\end{document}